\documentclass[11pt]{article}
\usepackage{fullpage, amsmath, amsfonts, amssymb, amsthm, natbib, bm}
\usepackage[toc,page]{appendix}
\usepackage{bbold} 
\usepackage{graphicx, algorithm, algorithmicx, enumerate, enumitem, color, hyperref}
\usepackage[noend]{algpseudocode}  
\usepackage{float}
\usepackage{bbm}
\usepackage{caption}

\newcommand{\real}{\mathbb R}
\def\st{\mathrm{\quad s.t.\quad}}

\def\Cov{\mathrm{Cov}}
\def\E{\mathrm{E}}
\def\Var{\mathrm{Var}}
\def\Prob{\mathrm{Pr}}

\def\diag{\operatorname{diag}}

\def\supp{\operatorname{supp}}
\def\iid{\operatorname{i.i.d.}}

\def\BigO{\mathcal O}
\def\onevec{\mathbf{1}}
\def\iden{\mathbf{I}}

\def\Event{\mathcal{E}}
\def\0{\mathbf{0}}
\def\r{\mathbf{r}}
\def\y{\mathbf{y}}
\def\z{\mathbf{z}}

\def\X{\mathbf{X}}

\def\W{\mathbf{W}}
\def\Z{\mathbf{Z}}

\def\C{\mathbf{C}}

\def\A{\mathcal A}
\def\T{\mathcal T}
\def\I{\mathcal I}

\def\x{\mathbf{x}}
\def\1{\mathbb{1}}

\def\tb{\beta^\ast}
\def\tg{\gamma^\ast}
\def\bvarepsilon{\boldsymbol{\varepsilon}}
\def\indi{\mathbbm{1}}
\def\top{\mathrm{top}}
\def\subw{\text{sub}\text{-Weibull}}

\newtheorem{theorem}{Theorem}
\newtheorem{lemma}[theorem]{Lemma}

\newtheorem{remark}[theorem]{Remark}
\newtheorem{corollary}[theorem]{Corollary}
\newtheorem{definition}[theorem]{Definition}

\newcommand{\vertiii}[1]
{{\left\vert\kern-0.25ex\left\vert\kern-0.25ex\left\vert #1
\right\vert\kern-0.25ex\right\vert\kern-0.25ex\right\vert}}

\newcommand{\norm}[1]{\left\lVert#1\right\rVert}
\newcommand{\snorm}[1]{\|#1\|}

\DeclareMathOperator*{\argmax}{arg\,max}
\DeclareMathOperator*{\argmin}{arg\,min}

\graphicspath{ {./} }

\date{}

\title{Reluctant Interaction Modeling}
\author{Guo Yu\thanks{Department of Statistics and Applied Probability, University of California Santa Barbara, Santa Barbara, CA 93110, \href{mailto:guoyu@ucsb.edu}{guoyu@ucsb.edu}} \and Jacob Bien\thanks{Data Sciences and Operations, Marshall School of Business, University of Southern California, Los Angeles, CA 90089, \href{mailto:jbien@usc.edu}{jbien@usc.edu}} \and Ryan Tibshirani\thanks{Department of Statistics, University of California Berkeley, Berkeley, CA 94720, \href{mailto:ryantibs@berkeley.edu}{ryantibs@berkeley.edu}} }

\begin{document}
\maketitle
\abstract{Including pairwise interactions between the predictors of a regression model can produce better predicting models. However, to fit such interaction models on typical data sets in biology and other fields can often require solving enormous variable selection problems with billions of interactions. The scale of such problems demands methods that are computationally cheap (both in time and memory) yet still have sound statistical properties. Motivated by these large-scale problem sizes, we adopt a very simple guiding principle: One should prefer main effects over interactions if all else is equal.
This ``reluctance'' to fit interactions, while reminiscent of the hierarchy principle for interactions, is much less restrictive.  We design a computationally efficient method built upon this principle and provide theoretical results indicating favorable statistical properties.  Empirical results show dramatic computational improvement without sacrificing predictive accuracy.  For example, the proposed method can solve a problem with 10 billion potential interactions with 5-fold cross-validation in under 7 hours on a single CPU.}


\section{Introduction}

In many prediction problems, 
modeling the response as a linear or even additive function of the features (i.e., as \textit{main effects}) is not sufficient for fully capturing the complexity of the relationship.
For example, many biological phenomena involve interactions among various behaviors, exposures, and genetic factors. 
Much recent work has therefore focused on predictive modeling and variable selection within the context of a two-way interaction model,
\begin{align}
  Y = X^T \beta^\ast + Z^T \gamma^\ast + \varepsilon,
  \label{model:true}
\end{align}
where $X \in \real^p$ is a $p$-dimensional random vector of main effects, $Z = (X_1^2, X_1 X_2, \dots, X_p^2) \in \real^{(p^2 + p)/2}$ is the random vector of all pairwise interactions of $X$, and $\varepsilon$ is a zero-mean noise random variable independent of $X$. Model \eqref{model:true} extends a typical linear model (in main effects $X$), and $\gamma^\ast$ characterizes how the pairwise interactions among features relate to the response.

For both good predictive performance and interpretability, it is common to search for a model with 
only a small number of nonzero components in $\tb$ and $\tg$.  For example, one might consider solving a lasso \citep{tibshirani1996regression} using all the main effects and the interactions (the so-called \textit{all pairs lasso}, APL).
In practice, APL quickly becomes infeasible to compute as $p$ gets large. Fitting APL with standard lasso solvers requires passing the whole augmented design matrix of main effects and interactions, which takes $\BigO(np^2)$ space. Moreover, even if we compute the interactions on the fly when solving APL,  state-of-the-art coordinate descent algorithms require multiple passes over all $\BigO(p^2)$ variables until convergence. 

However, beyond its computational problems, APL also poses certain conceptual problems.  In particular, it makes no fundamental distinction between main effects and interactions.  In practice, main effects are much simpler to interpret than interactions.  Also, there are $\BigO(p)$ times as many interactions than main effects, which means that APL has many more chances to choose an interaction than a main effect.

A standard approach to addressing both of these concerns is to adopt a hierarchy principle \citep{nelder1977reformulation, peixoto1987hierarchical, hamada1992analysis} that stipulates that an interaction can only be included in the model if one or more of its constituent main effects is also included.  
Many recent methods have incorporated the hierarchy assumption into a single optimization problem \citep{efron2004least, turlach2004discussion, zhao2009composite, yuan2009structured, choi2010variable, radchenko2010variable, schmidt2010convex, bien2013lasso, lim2015learning, haris2016convex, she2018group, hazimeh2020learning}. These methods become computationally challenging for larger problem sizes.  Other methods operate in multiple stages, exploiting the hierarchy assumption to attain greater computational efficiency \citep{wu2009genome, wu2010screen, hao2014interaction, shah2016modelling, hao2018model}. However, these multi-stage hierarchy methods often require that all nonzero elements of $\beta^\ast$ be detected in an early stage, which requires strong assumptions on the design and size of main effect coefficients.  

The hierarchy assumption, on which the aforementioned methods depend, might be too limiting in some problems \citep{culverhouse2002perspective}.
\citet{bien2013lasso} and \citet{hao2017note} provided justifications for the hierarchy assumption; however, these arguments are most convincing when the features are continuous. When features are categorical, the hierarchy assumption becomes more difficult to defend. 

The goal of this paper is to design a method for both continuous and categorical features that (a) is computationally scalable to large interaction models, (b) explicitly builds in the fact that main effects are to be preferred over interactions, and (c) does not assume hierarchy. 
To this end, we introduce in this paper a new method, which is extremely simple. At the same time, we show that the method allows one to fit sparse interaction models on scales not otherwise possible, while still preserving strong statistical properties.  More specifically, our contributions are as follows:

\begin{itemize}
  \item In large-scale interaction modeling, we propose to prioritize main effects over interactions given similar prediction performance. We emphasize that this ``reluctance'' to interactions is distinct from (although reminiscent of) the common hierarchy principle.
  \item 
    Motivated by this reluctance to interactions, we introduce a method called \textit{sprinter} (for sparse reluctant interaction modeling) that allows for interaction modeling without the hierarchy principle on unprecedented problem sizes without compromising statistical performance.  In particular, sprinter fits an interaction model with $2000$ main effects about 100 times faster than APL, and it fits a problem with about 10 billion interactions with 5-fold cross-validation in under 7 hours on a single CPU.
  \item We derive finite-sample theoretical properties of sprinter, including the computational complexity analysis and prediction error rate of sprinter.
\end{itemize}

We introduce the motivation and method in Sections \ref{sec:principle} and \ref{sec:method}, respectively.  In Section \ref{sec:theory}, 
we provide a theoretical analysis of the method, and in Section \ref{sec:num} we study the method's empirical performance.

\section{Motivation} \label{sec:principle}
One main reason for the computational and conceptual limitations of APL is that it treats the $p$ main effects and the $\BigO(p^2)$ interactions equivalently. The basic premise of our method is that we give specific preference to the main effects over interactions. Specifically, we propose to fit the response as well as possible using only a set of $\BigO(p)$ features based on main effects, i.e., univariate functions of the main effects, and then only include interaction terms for what cannot be captured by main effects in an additive manner. This motivates us to adopt the following simple guiding principle for large-scale interaction modeling, which we refer to as interaction reluctance:
\begin{quote} 
  {\em One should prefer main effects over interactions if all else is equal.}
\end{quote}
In particular, on the population level we seek univariate functions $f^\ast_j \in \mathcal{F}_j$, for some function spaces $\mathcal{F}_j$ and $j \in [p]$, that jointly minimize
\begin{align}
  \E [Y - \sum_{j = 1}^p f_j(X_j)]^2
  \nonumber
\end{align}
over $f_j \in \mathcal{F}_j$ for $j \in [p]$,
and then we seek a small number of interactions that capture the remaining signal in $Y - \sum_{j = 1}^p f^\ast_j(X_j)$ for $Y$ in \eqref{model:true}. In doing so, we give explicit preference to using main effects in capturing the signal in the response.
Specific examples of $\mathcal{F}_j$ include:
\begin{itemize}
  \item  Linear basis expansion in $X_j$: take $\bm{\theta}_j \in \real^{m_j}$, $\bm{g}_j(X_j) \in \real^{m_j}$ is a set of simple basis functions, and
    \begin{align}
      f_{j}\left(X_j\right) =  \bm{\theta}_{j}^T \bm{g}_{j}(X_j).
      \label{eq:basisexpansion}
    \end{align}
 For computational consideration, $m_j = O(1)$ is needed so that in total only $O(p)$ features derived from main effects are used.
 For simplicity of presentation, we will focus throughout this paper on the case where $\bm{g}_j = \{ X_j\}$ or $\bm{g}_j = \{X_j, X_j^2\}$. 
  \item Regression trees: $f_j(X_j)$ could be a regression tree on any set of derived features from $X_j$.
  \item Neural networks: $f_j(X_j)$ could be a neural network with any set of derived features from $X_j$ that are used for the first layer nodes.
\end{itemize}


In the interaction model \eqref{model:true} that we assumed throughout this paper, the additive model introduced above is used as a computationally-efficient approximation that encodes the reluctance principle.
Why would the proposed reluctance principle be effective in interaction modeling? In many cases, main effects (and) or simple univariate functions of main effects $f_j(X_j)$ are able to act as useful handles in approximating interactions.
\begin{figure}[ht]
  \centering
  \captionsetup{width=.8\textwidth}
  \includegraphics[width = .6\textwidth]{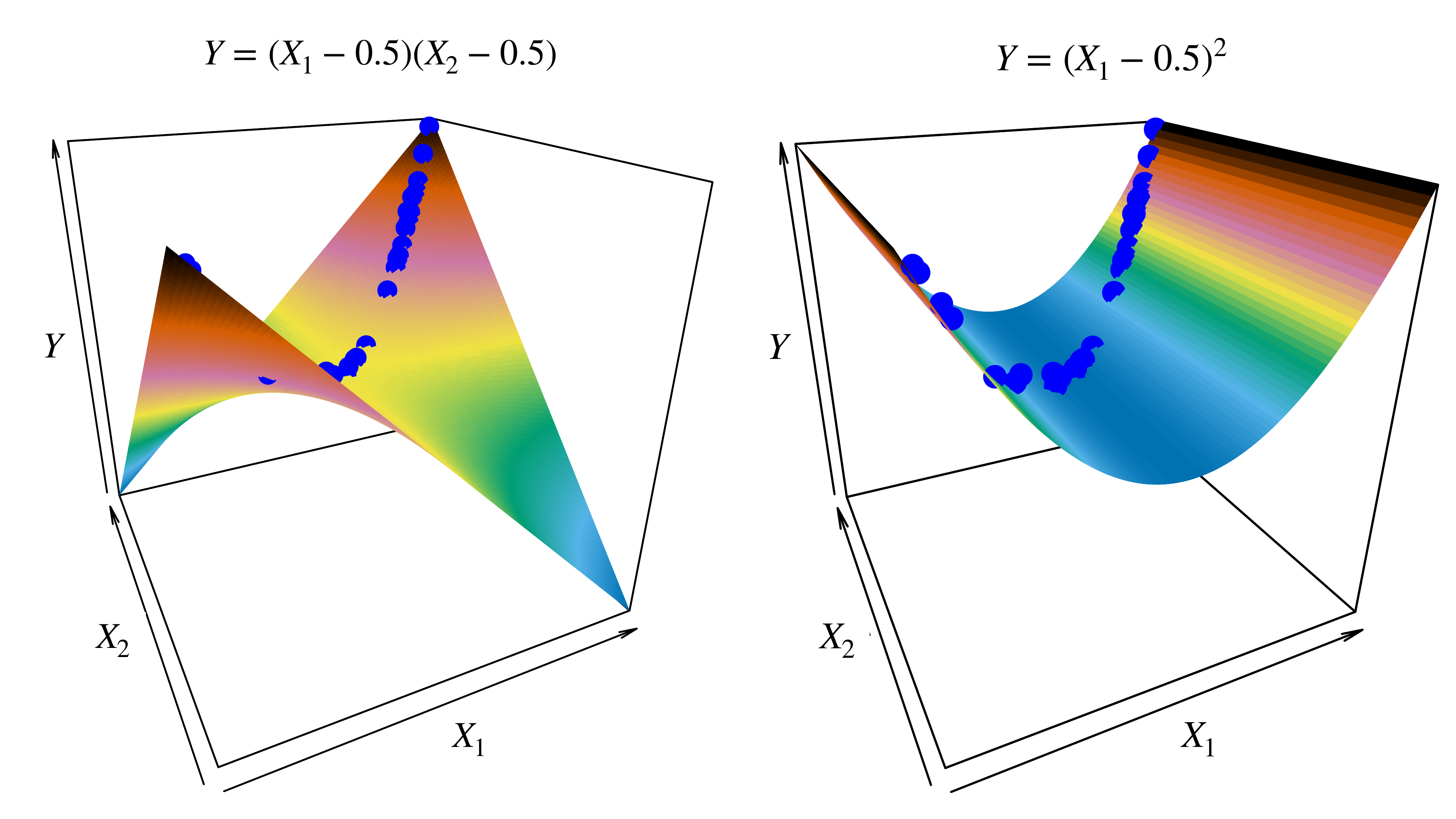}
  \caption{An example of an interaction that can be well approximated by the squared terms of either main effects when the two main effects are highly correlated. While the two surfaces are quite different, the perfect collinearity between $X_1$ and $X_2$ mean that all the observations (blue dots) fall along a line in which the interaction surface (on the left panel) can be perfectly explained by the squared-term surface involving only $X_1$ (on the right panel).}
  \label{fig:example}
\end{figure}
Figure~\ref{fig:example} shows an example where the interaction between two main effects can be very well approximated by either squared effect when these two constituent main effects are highly correlated. 
Given the widespread occurrence of feature correlation, especially in high-dimensional settings, our focus on $\bm{g}_j = \{X_j, X_j^2\}$ in \eqref{eq:basisexpansion} is particularly relevant in practice.
In a more general setting where specific interactions may be strongly correlated with linear combinations of the univariate functions of main effects $f^\ast_1(X_1), \ldots, f^\ast_p(X_p)$, we could get equivalently predictive models without using that interaction.

Such logic is not restricted to continuous predictors. Consider a simple model $Y = X_1 + X_2 + X_1 X_2$, where $X_1 = \indi_A$ and $X_2 = \indi_B$ are the indicator variables of events $A$ and $B$, respectively. Suppose further that with high probability $A \subseteq B$, so that $X_1 X_2 = \indi_{A} \indi_{B} = \indi_{A \cap B} \approx \indi_{A}$; hence, the main effect $X_1 = \indi_{A}$ can be used in place of the interaction $X_1 X_2$, i.e., $Y \approx 2 X_1 + X_2$. This main-effects-only explanation of $Y$ is simpler to understand and yet explains $Y$ nearly as well as the original model with the interaction term.

Leaning on main effects more heavily than interactions is advantageous for at least two reasons.
First, main effects are easier to interpret than interactions. 
Thus when presented with two models that predict the response equivalently, we should favor the one that relies on fewer interactions.
When putting forward a regression model with interactions, there is
an implication that the included interactions were needed. Second, we will show in this paper that prioritizing main effects (or simple functions of main effects such as squared terms) can lead to great computational savings both in terms of time and memory.  The key reason for these savings is that when $p$ is large, the total number of main effects is far smaller than the number of interactions.

We emphasize that the proposed reluctance to interactions is different from the well-known hierarchy assumption. While both simplify the search for interactions by focusing on certain main effects, our basic premise does not explicitly tie an interaction to its corresponding main effects. For example, an interaction $X_3 X_4$ could be highly correlated with a linear combination of $f^\ast_1(X_1)$ and $f^\ast_2(X_2)$, which may lead us to exclude $X_3 X_4$. On the other hand, the interaction $X_3 X_4$ will be kept if it cannot be explained by any additive functions of main effects. This logic is very different from the logic used in the hierarchy assumption.

\section{The sprinter algorithm} \label{sec:method}
In this section, we describe a new method, called sprinter, that is motivated by the idea of being reluctant to include interactions. As mentioned in Section \ref{sec:principle}, for the rest of this paper we focus on the special example where $f_j = \theta_j X_j$ or $f_j = \theta_{j1} X_j + \theta_{j2} X_j^2$.
The proposed method, for the case where $f_j(X_j) = \theta_j X_j$ in \eqref{eq:basisexpansion}, is given in Algorithm \ref{alg:first} and has three steps:
\begin{itemize}
  \item In Step 1, we fit a lasso (or any other regression method) of the response as well as possible using only main effects. This step purposely gives preference to main effects, motivated by the idea of prioritizing main effects over interactions described in Section \ref{sec:principle}.
  \item In Step 2, we perform a single pass over all interactions to identify interaction signal that was not captured in Step 1. Because each of the $\BigO(p^2)$ interactions is only computed and used once, this step requires far less time and memory than APL, which requires repeated passes over all interactions.
  \item In Step 3, we fit a lasso (or any other user-specified regression method) of the residual from Step 1 on all main effects and the interactions that were selected in Step 2. Depending on the screening criterion in Step 2, the total number of variables in Step 3 can be far smaller than $\BigO(p^2)$, leading to large computational gains over APL.
\end{itemize}

The lasso in Step 1 and Step 3 could be substituted by other regression methods. We choose the lasso as an example for subsequent analysis.
Throughout the paper, we let $q=(p^2 + p)/2$ be the total number of pairwise interactions between $p$ main effects.
On a sample level, we denote $\X = \left( \X_1,\dots, \X_p \right) \in \real^{n\times p}$ as the design matrix with 
each column $\X_j \in \real^n$ consisting of all observations of variable $X_j$ (for $j = 1,\dots,p$).
Similarly $\Z \in \real^{n \times q}$ is the sample matrix of $Z = (X_1^2, X_1 X_2, \dots, X_{p-1} X_p, X_p^2)$, and $\y \in \real^n$ is the response vector. 
We let $\overline{\mathrm{cor}}(\cdot, \cdot)$ and $\overline{\mathrm{sd}}(\cdot)$ stand for the sample correlation and the sample standard deviation respectively.

{\centering
  \begin{minipage}{\linewidth}
    \begin{algorithm}[H]
      \caption{sprinter (a lasso example)}
      \begin{algorithmic}[100]
        \Require Main effects $\X \in \real^{n \times p}$, interactions $\Z \in \real^{n \times q}$, response $\y \in \real^n$, $\eta > 0$
        \State \textbf{Step 1}:
        \State Fit a lasso of the response $\y$ on $\X$.
        \State Compute the residual $\r = \y - \X \hat{\theta}$.
        \State \textbf{Step 2}:
        \State For a tuning parameter $\eta$, screen based on the residual:
        \begin{align}
          \hat \I_\eta = \left\{\ell \in [q]: \overline{\mathrm{sd}}(\r)| \overline{\mathrm{cor}}\left( \Z_\ell,  \r \right)| > \eta \right\},
          \label{eq:Ihat}
        \end{align}
        \State \textbf{Step 3}:
        \State Fit a lasso of the residual $\r$ on $\X$ and $\Z_{\hat{\I}_\eta}$:
        \begin{align}
          (\hat{\xi}, \hat{\varphi}) \in \argmin_{\xi \in \real^p, \varphi \in \real^{|\hat{\I}_\eta|}}
          \left\{ \frac{1}{2n} \norm{\r - \X\xi - \Z_{\hat{\I}_\eta} \varphi}_2^2 + \lambda \left( \norm{\xi}_1 + \norm{\varphi}_1 \right)  \right\}.
          \label{eq:step3}
        \end{align}
        \Return $(\hat{\theta} + \hat{\xi}, \hat{\varphi})$
        \State For a new observation $\x$, the predicted response would be $\x^T (\hat{\theta} + \hat{\xi}) + \z_{\hat{\I}_\eta}^T \hat{\varphi}$.
      \end{algorithmic}
      \label{alg:first}
    \end{algorithm}
  \end{minipage}
}

For the case where $f_j = \theta_{j1} X_j + \theta_{j2} X_j^2$, we simply move the squared effects from $\Z$ to $\X$ in Algorithm \ref{alg:first}.
For more general choices of $f_j$ in \eqref{eq:basisexpansion}, $\X$ could be replaced by any design matrix of $\BigO(p)$ variables derived from main effects (in such a case, Step 2 would still only consider interactions between the original $p$ main effects).
For choices of $f_j$ other than \eqref{eq:basisexpansion}, one can simply change the fitting procedures in Step 1 and Step 3. 
Step 2 can be considered as a \textit{sure independence screening} \citep[SIS;][]{fan2008sure, barut2015conditional} of all interactions using the residual from Step 1.

These first two steps are built around the idea of being reluctant to include interactions. Given a set of highly correlated variables, the lasso tends to select just one of them. Thus, if an interaction is highly correlated with one or more main effects, APL may very well select the interaction. By contrast, sprinter explicitly prioritizes the main effects (in Step 1). An interaction will only be selected (in Step 2) if it can capture something in the signal that the main effects cannot. Finally, Step 3 jointly models the residual using main effects and the selected interactions, and the final model output is the combination of model fits in Step 1 and Step 3.

In practice, each step in Algorithm \ref{alg:first} involves a tuning parameter. In particular, there is a tuning parameter value triplet $(\lambda_1, \eta, \lambda_3)$, where $\lambda_1$ and $\lambda_3$ are tuning parameters for the lasso in Step 1 and Step 3, respectively, and $\eta$ is for the screening in Step 2.
The optimal value of these tuning parameters is usually unknown.
One approach would be to select the tuning parameter triplet using a three-dimensional cross-validation, which could be computationally prohibitive. 
As an alternative, we consider screening in Step 2 using
\begin{align}
  \hat \I_m^{\top} = \left\{\ell \in [q]: |\overline{\mathrm{cor}} \left( \Z_\ell, \r \right)| \text{ is among the } m \text{ largest} \right\}.
  \label{eq:Ihat_m}
\end{align}
This top-$m$ approach is standard in screening based variable selection methods \citep{fan2008sure, barut2015conditional} and large-scale interaction modeling approaches \citep{2016arXiv160508933F, niu2018interaction}. Popular choices of values of $m$ include $n$ and $\lceil n / \log (n) \rceil$. In Section \ref{sec:num}, we find that using $m = \lceil n / \log (n) \rceil$ consistently yields good performance, and cross-validating the value of $m$ has a rather small effect.
The scaling of $\overline{\mathrm{sd}}(\r)$ is needed to facilitate the theoretical analysis of the scaling of $\eta$ (see, e.g., Theorem \ref{thm:main} and Theorem \ref{thm:screening}). Clearly, when the top-$m$ approach is used, the scaling of $\overline{\mathrm{sd}}(\r)$ is unnecessary.

With a fixed value of $m$ (and equivalently, $\eta$), we only need a two-dimensional cross-validation of the lasso tuning parameter pair $(\lambda_1, \lambda_3)$, which is computationally viable due to the efficient path algorithms \citep{friedman2010regularization} for solving lasso problems in Step 1 and Step 3.

\subsection{A computationally efficient variant} \label{subsec:sprinter1cv}
As we will see in Section \ref{sec:theory}, Step 1 only needs to yield a good residual to guarantee the performance of Algorithm \ref{alg:first}. Therefore, we also consider the following computationally efficient variant of sprinter:

\textbf{sprinter(1cv)}: in Algorithm \ref{alg:first}, we first select the tuning parameter in Step 1 before proceeding to subsequent steps.

In Section \ref{sec:num}, we find that in most settings, \texttt{sprinter(1cv)} is doing only slightly worse than sprinter with two-dimensional cross-validation of $(\lambda_1, \lambda_3)$, with the benefit of faster computation.
However, there are situations in which \texttt{sprinter(1cv)} would substantially fall short. For example, when there is no main effects signal, \texttt{sprinter(1cv)}'s Step 1 tends to under-penalize whereas \texttt{sprinter}, which simultaneously cross-validates $(\lambda_1, \lambda_3)$, tends to pick a very large value of $\lambda_1$ and thus avoids this problem.

\subsection{Computation} \label{subsec:computation}
With a value of $m \leq n$, the required computation in both Step 1 and Step 3 are no worse than fitting a lasso with $p + n$ features. However, the major computational burden lies in Step 2, where $\BigO(p^2)$ sample correlations are computed. It is thus essential for this step to be implemented as efficiently as possible, both in terms of computational time and storage.

We compute the sample correlation between each interaction and the residual from Step 1 on the fly. In the meantime, a min-heap is used to keep the index pairs of the interactions that attain the $m$ largest sample correlations. This ensures that we do not have to store $\BigO(p^2)$ elements. The time complexity of Step 2 is thus $\BigO(p^2 (n + \log m))$, where $\BigO(p^2 \log m)$ is for maintaining the $m$ largest elements in a data stream of size $\BigO(p^2)$ by using a min-heap of size $m$ \citep{cormen2009introduction}. 
Note that the whole algorithm only requires $\BigO(n(p + m))$ storage. Various other data structures could be used to achieve similar computational and storage complexity. Step 2 could also be computed approximately using an equivalence with the correlation screening and closest-pair problems \citep{shamos1975closest, agarwal1991euclidean}, which could be solved approximately using locality sensitive hashing \citep{rajaraman2011mining, 2016arXiv161005108T} with a much improved computational complexity that is subquadratic in $p$.

\subsection{Related methods}
Our method is not alone in dropping the hierarchy assumption \citep[see, e.g.,][]{2016arXiv161005108T, 2016arXiv160508933F, niu2018interaction, 2018arXiv180107785R, wang2019penalized, zhou2019bolt}.
\textit{Interaction pursuit} (IP) operates in two stages, first seeking a subset of the original $p$ variables that are involved in the nonzero interactions and then restricting attention to interactions between these selected variables \citep{2016arXiv160508933F}. This method is efficient and can be quite effective.  Like multi-stage hierarchy methods and unlike our method, IP's success hinges on successful screening in the first step. Screening is easiest when the interactions are concentrated among a small set of original variables.  The most challenging situation for this method is when there is no such concentration of interactions over a small set of original variables.
Other screening-based methods exist. \citet{niu2018interaction} select interactions based on the partial correlation between the response and each interaction, with the corresponding two main effects adjusted. 
\citet{2018arXiv180107785R} screens interactions based on the three-way joint cumulant between the response and two main effects that make up an interaction. While these two methods account for the exact two main effects when selecting an interaction, our proposal is more general in that it will only select an interaction that cannot be explained by any linear combinations of all main effects.
Furthermore, our method is accompanied by finite sample theoretical guarantees, while such guarantees are not currently available for the methodology in \citet{niu2018interaction} and \citet{2018arXiv180107785R}.
\citet{2016arXiv161005108T} consider a randomized algorithm that solves each step of APL approximately by solving a closest-pair problem. By doing so, they show that the computational complexity of their method is sub-quadratic in $p$. 
Our method, while still having the same worst-case time complexity as that of APL, appears to be as fast as \citet{2016arXiv161005108T} in practice, and we find in our experiments it gives better predictive performance.

\section{Theoretical analysis} \label{sec:theory}
Model \eqref{model:true} expresses the signal in terms of a main effects signal term, $X^T \beta^\ast$, and an interactions signal term, $Z^T\gamma^\ast$. If $X$ and $Z$ were uncorrelated, this would be a unique decomposition. However, as demonstrated in Section \ref{sec:principle}, there can be ``overlap'' between these two signal terms. Let $X^T \vartheta^\ast$ be the part of $Z^T \gamma^\ast$ that can be explained by a linear combination of $X$, i.e., 
\begin{align}
  \vartheta^\ast := \argmin_{\vartheta\in\real^p} \Var\left( Z^T \gamma^\ast - X^T \vartheta \right)=\Cov \left( X \right)^{-1} \Cov \left( X, Z \right) \gamma^\ast= \Sigma^{-1}\Phi\gamma^\ast,
  \label{eq:vartheta}
\end{align}
where we denote $\Sigma = \Cov(X)\in\real^{p\times p}$ and $\Phi = \Cov(X, Z)\in\real^{p\times q}$. 
We can then write \eqref{model:true} as
\begin{align}
  Y = X^T \theta^\ast  + W^T \gamma^\ast + \varepsilon,
  \label{eq:true_repa}
\end{align}
where $\theta^\ast = \beta^\ast + \vartheta^\ast$, and
\begin{align}
  W := Z - \Phi^T \Sigma^{-1} X
  \nonumber
\end{align}
is the ``pure'' interaction effects that cannot be captured by linear combinations of $X$, with $\Cov(X, W) = 0$. We denote the covariance of the pure interactions as
\begin{align}
  \Omega := \Cov(W) = \Cov\left( Z, W \right) = \Cov\left( Z \right)- \Cov \left( Z, X \right) \Sigma^{-1} \Phi = \Psi - \Phi^T \Sigma^{-1} \Phi,
  \label{eq:Omega}
\end{align}
where $\Psi = \Cov(Z)\in\real^{q\times q}$.
We note that $X^T \theta^\ast$ is aligned with \textit{the best population linear approximation} to $Y$ in \citet{buja2019models}.
By fitting $Y$ using only linear combinations of $X$, we fit a misspecified model because the pure interaction $W$ is ignored. We will see in the following analysis that the zero covariance structure between $X$ and $W$ is helpful in reducing the effects of this model misspecification.  Actually, the benefit from the ``orthogonality'' between main effects and interactions is also observed in \citet{hao2014interaction}, where $X$ is assumed to follow a zero-mean symmetric distribution. In such a case, we have that $\Phi = \Cov(X, Z) = 0$, which implies that $\vartheta^\ast = 0$ and $W = Z$. Our method does not require the symmetry of the distribution of main effects and thus allows for more general covariance structure between main effects $X$ and interactions $Z$. 

We also note that $X$ in \eqref{eq:true_repa} can be generalized to be a random vector containing main effects and simple functions of main effects, e.g., the squared effects, or general univariate nonlinear functions of main effects $f_j(X_j)$'s (as in Section \ref{sec:principle}).

To study sprinter's performance, we will need to make assumptions about the tail behavior of the features. If main effects are sub-Gaussian, then their interactions are known to be sub-exponential. However, the analysis of interaction modeling usually involves the product of more than two main effects (e.g., the product of a main effect and an interaction), which has heavier tails than sub-exponential random variables.
The following definition \citep[][Definition 2.2 and 2.4]{kuchibhotla2018moving} will therefore be useful.
\begin{definition}[$\subw(\nu)$ random variable/vector] \label{def:sb}
  A random variable $U$ is a sub-Weibull random variable of order $\nu > 0$, i.e., $\subw(\nu)$, if 
  \begin{align}
    \snorm{U}_{\psi_\nu} = \inf \left\{ \zeta > 0: \E \left[ \exp \left( \frac{|U|^\nu}{\zeta^\nu}  \right) \right] \leq 2 \right\} < \infty,
    \label{eq:sb}
  \end{align}
  where $\snorm{U}_{\psi_\nu}$ is the Orlicz norm of $U$. A random vector $V \in \real^p$ is a $\subw(\nu)$ random vector if $c^T V$ is a $\subw(\nu)$ random variable for any constant vector $c \in \real^p$.  Furthermore, we define $\snorm{V}_{\psi_\nu} = \sup_{\snorm{c}_2 = 1} \snorm{c^T V}_{\psi_\nu}$.
\end{definition}
The notion of $\subw(\nu)$ generalizes the definition of sub-Gaussian, which is $\subw(2)$, and sub-exponential, which is $\subw(1)$. In this paper, we are primarily interested in the cases where $\nu < 1$. In particular, the product of three and four sub-Gaussian main effects, which as shown in Appendix \ref{app:subweibull}, are $\subw(2/3)$ and $\subw(1/2)$ respectively.
In Appendix \ref{app:subweibull}, we also give a set of concentration inequalities for these heavy tailed random variables.

In the theoretical analysis of our method, we make the following assumptions:
\begin{enumerate}[label=\textbf{A\arabic*}]
  \item \label{a:distn} We have $n$ independent samples from \eqref{model:true}, where $X = (X_1, \dots, X_p)$ follows a zero-mean sub-Gaussian distribution with covariance matrix $\Sigma$ and sub-Gaussian norm $\snorm{X}_{\psi_2}$, and $\varepsilon$ is zero-mean sub-Gaussian noise independent of $X$ with $\Var(\varepsilon) = \sigma^2$.
  \item \label{a:dim} We assume that $\kappa \log p \leq \sqrt{n}$ with some constant $\kappa > 1$.
\end{enumerate}

Assumption \ref{a:distn} is a very general distributional assumption on the random design; unlike other methods in high-dimensional interaction modeling, we do not require the distribution of $X$ to be symmetric. 
Assumption \ref{a:dim} is a standard assumption for pairwise interaction screening consistency \citep[see, e.g., Assumption (C4) of][]{hao2014interaction}. 
This is more stringent than the standard sample size requirement in main effects screening \citep[see, e.g., ][]{fan2008sure}, which requires $\log p = \BigO(n)$.
This is because interactions concentrate around their population means more slowly due to their heavier tails. 

Ultimately we want to characterize the prediction performance and computational complexity of sprinter. In particular, the prediction error is characterized by 
\begin{align}
  \frac{1}{2n} \norm{\X \theta^\ast + \W \gamma^\ast - \X \hat{\theta} - \X \hat{\xi} - \Z_{\hat{\I}_\eta} \hat{\varphi}}_2^2,
  \nonumber
\end{align}
where $\hat{\theta}$ is the Step 1 coefficient estimate of $\theta^\ast$ in \eqref{eq:true_repa}, $(\hat{\xi}, \hat{\varphi})$ is the coefficient estimate from \eqref{eq:step3} given in Step 3 of Algorithm \ref{alg:first}, and
$\hat{\I}_{\eta}$ is the output of Step 2 in \eqref{eq:Ihat} with tuning parameter $\eta$.

Clearly, by taking $\eta = 0$, we have $\hat{\I}_\eta = [ q ]$. In this case, there is no computational gain over APL because all interactions will be considered in Step 3.
On the other extreme, if $\eta = \infty$, then $\hat{\I}_\eta = \emptyset$ and the whole procedure ignores all of the pure interaction signal $W^T\gamma^\ast$.
Therefore, the success of sprinter hinges on capturing a small set $\hat{\I}_\eta$ that still captures enough of the pure interaction signal $W^T \gamma^\ast$.

We define the target set of interactions to recover as
\begin{align}
  \I(\alpha) \in \argmax_{\A \subseteq [q]}\left\{\min_{\ell \in \A} \left|\Psi_{\ell \ell}^{-1/2} \Cov\left(Z_\ell, W^T\gamma^\ast\right)\right| \st \snorm{(W_{\A^C}^T \gamma_{\A^C}^\ast)^2}_{\psi_{1/2}} \leq \alpha \right\}.
  \label{eq:I_alpha}
\end{align}
We first explain the constraint in \eqref{eq:I_alpha}.
Recall that $W^T\gamma^\ast$ is the part of the interaction signal that cannot be explained by linear combinations of main effects. 
In Appendix \ref{app:subweibull} we show that $(W^T \gamma^\ast)^2$ is a $\subw(1/2)$ random variable.
For any value $\alpha \geq 0$, we call a set $\A \subseteq [q]$ of interactions {\em $\alpha$-important} if $\snorm{(W^T_{\A^C} \gamma^\ast_{\A^C})^2}_{\psi_{1/2}} \leq \alpha$.
\begin{remark}
  Lemma \ref{lem:exp} (in Appendix A) shows that $\E[(W^T_{\A^C} \gamma^\ast_{\A^C})^2] \leq 4 \snorm{(W^T_{\A^C} \gamma^\ast_{\A^C})^2}_{\psi_{1/2}} \leq 4\alpha$. Therefore, an $\alpha$-important set of interactions captures all but $4 \alpha$ of the pure interaction signal.
\end{remark}
Here, $\alpha$ is a theoretical tuning parameter controlling the size of the target (population-level) screening set $\I(\alpha)$ in Step 2, thereby determining the trade-off between computation and prediction error. In typical interaction modeling, the goal would be to recover $\supp(\gamma^\ast)$, which is a $0$-important set; however, by taking larger $\alpha$, we can reduce the size of the target interaction screening set, improving computation with controlled cost to prediction error.
In our reluctant interaction selection framework, we do not care about recovering a set of interactions $\mathcal{B} \subseteq \supp(\gamma^\ast)$ if $\snorm{(W_{\mathcal{B}}^T \gamma^\ast_{\mathcal{B}})^2}_{\psi_{1/2}}$ (and thus $\E[(W_{\mathcal{B}}^T \gamma^\ast_{\mathcal{B}})^2]$) is small.  For example, if $Z_\ell$ can be perfectly explained by a linear combination of main effects, then we do not wish to select interaction $\ell$ even if $\gamma_\ell^\ast \neq 0$.  A strength of our theoretical results is that they are in terms of general $\alpha$, thus making explicit the trade-off between computational efficiency and prediction accuracy.

From \eqref{eq:I_alpha} we know that $\I(\alpha)$ is an $\alpha$-important set.
Yet for any $\alpha \geq 0$, the $\alpha$-important set is not necessarily unique.
If Step 1 does a good job of capturing all the signal from the main effects, i.e., $\X \theta^\ast \approx \X \hat{\theta}$, where $\X \hat{\theta}$ is the fitted response from Step 1, then $\r = \y - \X \hat{\theta}$ should essentially be the pure interaction signal $\W \gamma^\ast$ (with noise).
Step 2 obtains $\hat{\I}_\eta$ by including all the interactions whose (scaled) sample correlation with the residual is large enough, i.e., $\omega_\ell = \overline{\mathrm{sd}}(\r) |\overline{\mathrm{cor}}\left( \Z_\ell,  \r \right)| \geq \eta$ for some $\eta \geq 0$, where $\omega_\ell$ is a noisy proxy of the population signal strength 
$$
\omega^\ast_\ell = \Psi_{\ell \ell}^{-1/2} \Cov(Z_\ell, W^T\gamma^\ast).
$$
For any $\alpha$-important set $\A$ to be detectable, we require that the minimum signal strength $\min_{\ell \in \A} |\omega_\ell^\ast|$ in $\A$ is large enough to be differentiated from the noise. 
The target set of interactions $\I(\alpha)$ in \eqref{eq:I_alpha} is thus defined as the $\alpha$-important set that is most easily detected since it has the largest minimum signal strength, which we define as
\begin{align}
  \eta (\alpha) := \frac{1}{2} \min_{\ell \in \I(\alpha)} |\omega_\ell^\ast|.
  \label{eq:eta_alpha}
\end{align}

The following theorem shows the main theoretical properties of sprinter: it attains good prediction accuracy while being computationally efficient when the minimum signal strength is greater than a certain noise level.
We provide a roadmap to its proof in Section \ref{subsec:roadmap}, deferring the full proof of this and all other results to the appendix.
\begin{theorem} \label{thm:main}
  Define
  \begin{align}
    \eta^\ast =& K \left[  \left(\snorm{\diag(\Psi)^{-1/2}Z}_{\psi_1}\snorm{W^T \gamma^\ast}_{\psi_1} + \max_{\ell} |\omega_\ell^\ast| \right)  \frac{(\log p)^{3/4}}{n^{1/2}} \right. \nonumber\\
               & \left.+ \left(\sigma + \snorm{W^T\gamma^\ast}_{\psi_{1}} \right)^{1/2} \snorm{X}_{\psi_2}^{1/2} \frac{(\log p)^{1/4}}{n^{1/4}}+  \sigma \frac{(\log p)^{1/2}}{n^{1/2}} \right],
               \label{eq:eta_star}
  \end{align}
  to be the noise level, where $K > 0$ is a constant.
  Under Assumption \ref{a:distn} and \ref{a:dim}, for any $\alpha \geq 0$, under the signal strength condition that $\eta(\alpha) \geq \eta^\ast$, and taking
  \begin{align}
    \lambda_0 = C \left(\sigma + \snorm{W^T\gamma^\ast}_{\psi_{1}} \right) \snorm{X}_{\psi_2} \sqrt{\frac{\log p}{n}}, \qquad
    \lambda = C_1 \sigma \max \left( \snorm{X}_{\psi_2}, \snorm{Z}_{\psi_1} \right)\sqrt{\frac{ \log p}{n}}
    \label{eq:lambda_cor1}
  \end{align}
  as the tuning parameter value in Step 1 and Step 3 respectively, then for any $\eta \in [\eta^\ast, \eta(\alpha)]$, sprinter achieves:
  \begin{enumerate}
    \item (Screening property, implying computational efficiency)
      \begin{align}
        \I(\alpha) \subseteq \hat{\I}_\eta \quad \text{and} \quad |\hat{\I}_{\eta}| \leq 4{\eta}^{-2}\lambda_{\max}\left( \diag(\Psi)^{-1/2}\Omega\diag(\Psi)^{-1/2} \right) \Var(W^T\gamma^\ast) 
        \label{eq:efficiency}
      \end{align}
    \item (Prediction error rate)
      \begin{align}
&\frac{1}{2n} \norm{\X \theta^\ast - \X \hat{\theta} + \W \gamma^\ast - \X \hat{\xi} - \Z_{\hat{\I}_\eta} \hat{\varphi}}_2^2 \nonumber\\
        \leq & C_2 ( \snorm{W^T \gamma^\ast}_{\psi_1} \snorm{X}_{\psi_2}\sqrt{\frac{\log p}{n}} + \lambda )\snorm{\theta^\ast}_1 \nonumber\\
             &+ \inf_{\bar{\alpha} \geq \alpha} \left \{C_3 \bar{\alpha} + 4\lambda \left( \snorm{\Sigma^{-1} \Phi_{\I(\bar{\alpha})} \gamma^\ast_{\I(\bar{\alpha})}}_1 + \snorm{\gamma^\ast_{\I(\bar{\alpha})}}_1 \right) \right\}
        \label{eq:bound_pred}
      \end{align}
  \end{enumerate}
  with probability greater than $1 - 2 \exp \left( -n^{3/5} \right) - 2 p^{-1} - 16 p^{-2(\kappa^{1/3} - 1)}$, where $C_1, C_2$, and $C_3$ are positive constants.
\end{theorem}

The value of $\eta^\ast$ in \eqref{eq:eta_star} is the noise level in Step 2, which is the sum of three terms: 
the first term depends on the strength of the pure interaction signal.
The second term stems from the prediction error bound from Step 1, which fits a misspecified model since it ignores the pure interaction signal $W^T \gamma^\ast$ in \eqref{eq:true_repa}.
The last term depends on the error standard deviation $\sigma$. 
Actually from Theorem \ref{thm:lasso_step1} we see that the second term is a slow rate prediction error bound for Step 1. The results in Theorem \ref{thm:main} can thus be generalized by replacing the second term with a faster rate under stronger assumptions if lasso is still used in Step 1, or any other prediction error bound available to a generic method used in Step 1.
In such a generalization, only the second term in $\eta^\ast$ and the probability with which \eqref{eq:efficiency} and \eqref{eq:bound_pred} hold will be changed.

The result in \eqref{eq:efficiency} characterizes the size of retained interactions used in Step 3. In particular, it implies that if $\Var(W^T\gamma^\ast) = 0$, i.e., if all interaction signal can be explained by main effects, then $\hat{\I}_\eta = \emptyset$, and Step 3 is not needed at all.
Finally, the results in Theorem \eqref{thm:main} hold with probability tending to $1$ as $p \rightarrow \infty$ and $n \rightarrow \infty$.

Computationally, recall that Step 3 of sprinter is solving a lasso with $O(p + |\hat{\I}_{\eta}|)$ variables. If $|\hat{\I}_\eta| = o(p^2)$, then Step 3 of sprinter is computationally more efficient than APL because it solves a problem with a smaller number of variables. By the orthogonality between $X$ and $W$, we have $\Var(Y) = \Var(X^T \theta^\ast) + \Var(W^T\gamma^\ast) + \sigma^2$. Under a standard assumption that $\Var(Y) = \BigO(1)$ \citep[see, e.g.,][]{fan2008sure}, if the maximum singular value $ \lambda_{\max}\left( \diag(\Psi)^{-1/2}\Omega\diag(\Psi)^{-1/2} \right) = o({\eta^\ast}^{2}p^2)$, then \eqref{eq:size_Ihat} implies that $|\hat{\I}_{\eta}| = o(p^2)$ and thus sprinter is computationally more efficient than APL.

In Section \ref{subsec:example} we consider an example, where we explicitly write out the condition under which the condition $\eta(\alpha) \geq \eta^\ast$ holds for different values of $\alpha$.

\subsection{A roadmap to the proof of Theorem \ref{thm:main}} \label{subsec:roadmap}
In this section, we give a series of theoretical results derived in order to prove Theorem \ref{thm:main}.
As discussed in the previous section, the success of Step 3 depends on Step 2 achieving a type of screening property, i.e., it retaining all the important interactions in $\I(\alpha)$, and also that $|\hat{\I}_\eta|$ is not too large. The following theorem shows that if the signal strength condition that $\eta(\alpha) \geq \eta^\ast$ holds, then Step 2 yields a small screening set that contains $\I(\alpha)$.
\begin{theorem}[Screening property in Step 2] \label{thm:screening}
  Consider the event $$\Event_R = \left\{ n^{-1/2} \snorm{\X\theta^\ast - \X \hat{\theta}}_2 \leq R \right\}$$ for some prediction error rate $R$ of Step 1.
  Let
  \begin{align}
    \bar{\eta}^\ast = K \left[\left(\snorm{\diag(\Psi)^{-1/2}Z}_{\psi_1} \snorm{W^T \gamma^\ast}_{\psi_1} + \max_{\ell} |\omega_\ell^\ast| \right)  \frac{(\log p)^{3/4}}{n^{1/2}} + R +  \sigma \frac{(\log p)^{1/2}}{n^{1/2}} \right],
    \label{eq:eta_star_bar}
  \end{align}
  where $K > 0$ is a constant.
  Under Assumption \ref{a:distn} and \ref{a:dim}, for any $\alpha \geq 0$, if $\eta (\alpha) \geq \bar{\eta}^\ast$,
  then for any $\eta \in [\bar{\eta}^\ast, \eta(\alpha)]$,
  \begin{align}
    \I(\alpha) \subseteq \hat{\I}_{\eta} \quad \text{and} \quad |\hat{\I}_{\eta}| \leq 4{\eta}^{-2}\lambda_{\max}\left( \diag(\Psi)^{-1/2}\Omega\diag(\Psi)^{-1/2} \right) \Var(W^T\gamma^\ast) 
    \label{eq:size_Ihat}
  \end{align}
  holds with probability greater than $1 - 8 p^{-2(\kappa^{2/5} - 1)} - 2 p^{-1} - \Prob(\Event_R^C)$. 
\end{theorem}

As discussed in the previous section, for the set $\I(\alpha)$ to be recovered, the condition that $\eta(\alpha) \geq \bar{\eta}^\ast$ requires that the minimum signal strength $\eta(\alpha)$ should be stronger than $\bar{\eta}^\ast$, which, intuitively can be considered as the noise level of Step 2. This is similar to a ``$\beta$-min'' condition in the screening and variable selection consistency literature \citep[see, e.g.,][]{fan2008sure, wainwright2009sharp}.
The definition of $\bar{\eta}^\ast$ is more general than \eqref{eq:eta_star} in that it allows for a generic prediction error rate $R$ of Step 1. As a result, the probability with which \eqref{eq:size_Ihat} holds depends on $\Prob(\Event_R^C)$. As discussed earlier, applying a different method in Step 1 or a different rate using lasso (under stronger conditions) will result in a different $R$, which changes the bounds (and the corresponding probability) in Theorem \ref{thm:main}.

As with other methods in interaction screening, we note that the result in Theorem \ref{thm:screening} is less favorable (which is mostly reflected in a stricter sample size requirement \ref{a:dim}) than that of sure independence screening \citep{fan2008sure} in the main-effects-only case. This reveals an intrinsic challenge when dealing with interactions, namely that they have heavier tails than main effects. When we further assume that $X$ has a bounded distribution, Theorem \ref{thm:screening} can be much improved. The rate is still less good than the main-effects-only case, as it depends on the misspecified lasso fit in the first step---an expected caveat in a two-stage method.

\begin{remark}[The ``top-$m$'' strategy]
Both $\eta(\alpha)$ and $\bar{\eta}^\ast$ depend on some population quantities, and thus are not available in practice.
We thus adapt the same ``top-$m$'' strategy as \citet{fan2008sure} and \citet{2016arXiv160508933F} using $\hat{\I}^\top_m$ of \eqref{eq:Ihat_m}. 
Appendix \ref{app:screening_m} shows that Step 2 succeeds if we adapt the ``top-$m$'' approach under certain conditions.
\end{remark}

With Theorem \ref{thm:screening} we have shown \eqref{eq:efficiency} in Theorem \ref{thm:main}. To show \eqref{eq:bound_pred}, the following theorem first gives a deterministic bound on the prediction error of Step 3 if $\I(\alpha) \subseteq \hat{\I}_\eta$ holds.
\begin{theorem}[Prediction error in Step 3] \label{thm:prediction}
  Consider the event $$\Event_R = \left\{ n^{-1/2} \snorm{\X\theta^\ast - \X \hat{\theta}}_2 \leq R \right\}$$ for some prediction error rate $R$ of Step 1.
  For any $\alpha \geq 0$, suppose $\I(\alpha) \subseteq \hat{\I}_{\eta}$, and take 
  \begin{align}
    \lambda \geq \max \left( \frac{1}{n} \max_{1 \leq j \leq p} |\bvarepsilon^T \X_j|, \frac{1}{n} \max_{\ell \in \hat{\I}_{\eta}} |\bvarepsilon^T \Z_\ell| \right)
    \label{eq:lambda}
  \end{align}
   as the tuning parameter value in Step 3.
  On the event $\Event_R$, we have
  \begin{align}
    &\frac{1}{2n} \norm{\X \theta^\ast + \W \gamma^\ast - \X(\hat{\theta} + \hat{\xi}) - \Z_{\hat{\I}_\eta} \hat{\varphi} }_2^2 \nonumber\\
    \leq & R^2 + \frac{1}{n} \norm{\W_{\I(\alpha)^C} \gamma^\ast_{\I(\alpha)^C}}_2^2 + 2\lambda \left( \snorm{\Sigma^{-1} \Phi_{\I(\alpha)} \gamma^\ast_{\I(\alpha)}}_1 + \snorm{\gamma^\ast_{\I(\alpha)}}_1 \right).
    \label{eq:general_bound}
  \end{align}
\end{theorem}

The result is deterministic in that it does not require any probabilistic argument. Based on Theorem \ref{thm:prediction}, Theorem \ref{thm:main} then characterizes the scale of $\lambda$ and the corresponding probability that \eqref{eq:lambda} holds.

Finally, we derive a particular prediction error rate $R$ in Step 1 used in Theorem \ref{thm:screening}.
Although the proposed framework does not depend on a specific regression method in Step 1 for fitting the main effects, we take the lasso as an example. Recall that in Step 1 we are fitting a misspecified model, i.e., we treat $W^T \gamma^\ast + \varepsilon$ in \eqref{eq:true_repa} as the noise term and solve the following problem:
$$
\check{\theta} \in \argmin_\theta \left( \frac{1}{2n} \norm{\y - \X\theta}^2_2+\lambda\norm{\theta}_1 \right).
$$
The following theorem gives a prediction error rate for the main effects only lasso that is carried out in Step 1. 
\begin{theorem} [Prediction error in Step 1]
  \label{thm:lasso_step1}
  Suppose that Assumption \ref{a:distn} and \ref{a:dim} hold.  Take
  \begin{align}
    \lambda_0 = C \left(\sigma + \snorm{W^T\gamma^\ast}_{\psi_{1}} \right) \snorm{X}_{\psi_2} \sqrt{\frac{\log p}{n}},
    \nonumber
  \end{align}
  where $C > 0$ is a constant,
  then for any $\lambda \geq \lambda_0$,
  the following bound holds
  \begin{align}
    \frac{1}{n} \snorm{\X \check{\theta} - \X \theta^\ast}_2^2 \leq 4\lambda \snorm{\theta^\ast}_1
    \nonumber
  \end{align}
  with probability greater than $1 - 4 p^{-(\kappa^{1/3} - 1)}$.  
\end{theorem}
We see, through $\lambda$'s dependence on $\snorm{W^T \gamma^\ast}_{\psi_1}$, that the presence of the pure interaction signal leads to rates that could be less good than if no interactions were present. This is the price paid for fitting a misspecified model in Step 1. Also this prediction bound holds under Assumption \ref{a:dim}, which is a stricter sample size requirement due to dealing with the empirical process that involves interactions, which have heavier tails than main effects. 
This bound is aligned with the oracle inequality for the constrained lasso in the misspecified model (with a sub-Weibull(1) misspecification part).
Under stronger conditions (e.g., compatibility conditions on $\theta^\ast$), a faster prediction error rate in Step 1 could be derived.

\subsection{A Gaussian example} \label{subsec:example}
In this section we study the condition that $\eta(\alpha) \geq \eta^\ast$ required both in Theorem \ref{thm:main} and Theorem \ref{thm:screening} in the case where $X$ follows a Gaussian distribution. We defer the detailed computation to Appendix \ref{app:example}, where we also consider an example where $X$ is not symmetric. 

We let $\tau: [p]^2 \rightarrow [q]$ map the interaction between $X_j$ and $X_k$ to its corresponding index in $Z$, i.e., $Z_{\tau(j, k)} = X_j X_k$.
Consider the simple case where $X \sim N(\mathbf{0}, \Sigma)$, and there is only one true interaction, i.e., $\supp(\gamma^\ast) = \{\tau(1, 2)\}$. Without loss of generality, we assume that $\Sigma_{jj} = 1$ for all $j = 1, \dots, p$, so that for any pair of variables $X_j$ and $X_k$, their covariance $\sigma_{jk}$ equals their correlation coefficient $\rho_{jk}$. Recall from Section \ref{sec:principle} that in the Gaussian case, $\Phi = \0$, $W = Z$ and $\theta^\ast = \beta^\ast$.
We discuss two cases, depending on the size of $|\gamma_{\tau(1, 2)}^\ast|$:    
\begin{enumerate}
  \item (Strong interaction) Suppose the signal is strong in that 
    \begin{align}
      |\gamma^\ast_{\tau(1, 2)}| > r(n, q)^2,
      \nonumber
    \end{align}
    where
    \begin{align}
      r(n, q)^2 = 18 K^2 \snorm{\beta^\ast}_{1} \left( \frac{\log p}{n}  \right)^{1/2} + 12K  \snorm{\beta^\ast}_1^{1/2} \sigma^{1/2} \left( \frac{\log p}{n}  \right)^{1/4} + 12K \sigma \left( \frac{\log p}{n}  \right)^{1/2},
      \nonumber
    \end{align}
    and $K > 0$ is the constant in \eqref{eq:eta_star}.
    Appendix \ref{app:example} shows that $|\gamma^\ast_{\tau(1, 2)}| > r(n, q)^2$ is a sufficient condition under which $\eta(\alpha) \geq \eta^\ast$ holds for any $\alpha \geq 0$. 
  \item (Weak interaction) Suppose the signal strength is weak in that 
    \begin{align}
      |\gamma^\ast_{\tau(1, 2)}| \leq r(n, q)^2.
      \nonumber
    \end{align}
    Taking
    \begin{align}
      \alpha = 3 r(n, q)^4 = \BigO \left( \sigma \snorm{\beta^\ast}_1 \left( \frac{\log p}{n}  \right)^{1/2} \right),
      \label{eq:alpha}
    \end{align}
    we can show that
    $\I(\alpha) = \emptyset$.
    As a result, \eqref{eq:eta_alpha} implies that $\eta(\alpha) = \infty$, which is trivially greater than $\eta^\ast$. 
\end{enumerate}
Finally we note that the reason for condition $\eta(\alpha) \geq \eta^\ast$ to be satisfied regardless of the size of the interaction signal strength is because of our focus on prediction error in Theorem \ref{thm:main} and Theorem \ref{thm:screening} instead of the exact recovery of specific interactions.

\section{Numerical studies} \label{sec:num}
\subsection{Simulation studies: binary features} \label{subsec:binary}
We consider a simulation scenario with binary features in which some but not all interactions can be well approximated by main effects.  We generate $p$ binary features as follows: for a (perfect) binary tree of depth $d$, each leaf node is an independent Bernoulli($0.1$) random variable; the value of each non-leaf node is the maximum of the node values in its sub-tree, i.e., each non-leaf node represents an event that is the union of all the events represented by its children nodes. The total number of nodes in the tree is $p = 2^{d + 1} - 1$, and we consider these node values as main effects. 
This construction ensures that for any pair of main effects, they are either independent or else one is an ancestor of the other. The interaction between two binary features is simply the intersection of the two main effect events, so in this second case their interaction is simply the main effect corresponding to the descendant node. Figure~\ref{fig:btree} shows the binary tree (of depth $5$), where each node represents a main effect, and the node color represents the success probability of the corresponding Bernoulli random variable.
\begin{figure}[ht]
  \centering
  \captionsetup{width=.8\textwidth}
  \includegraphics[width = .8\textwidth]{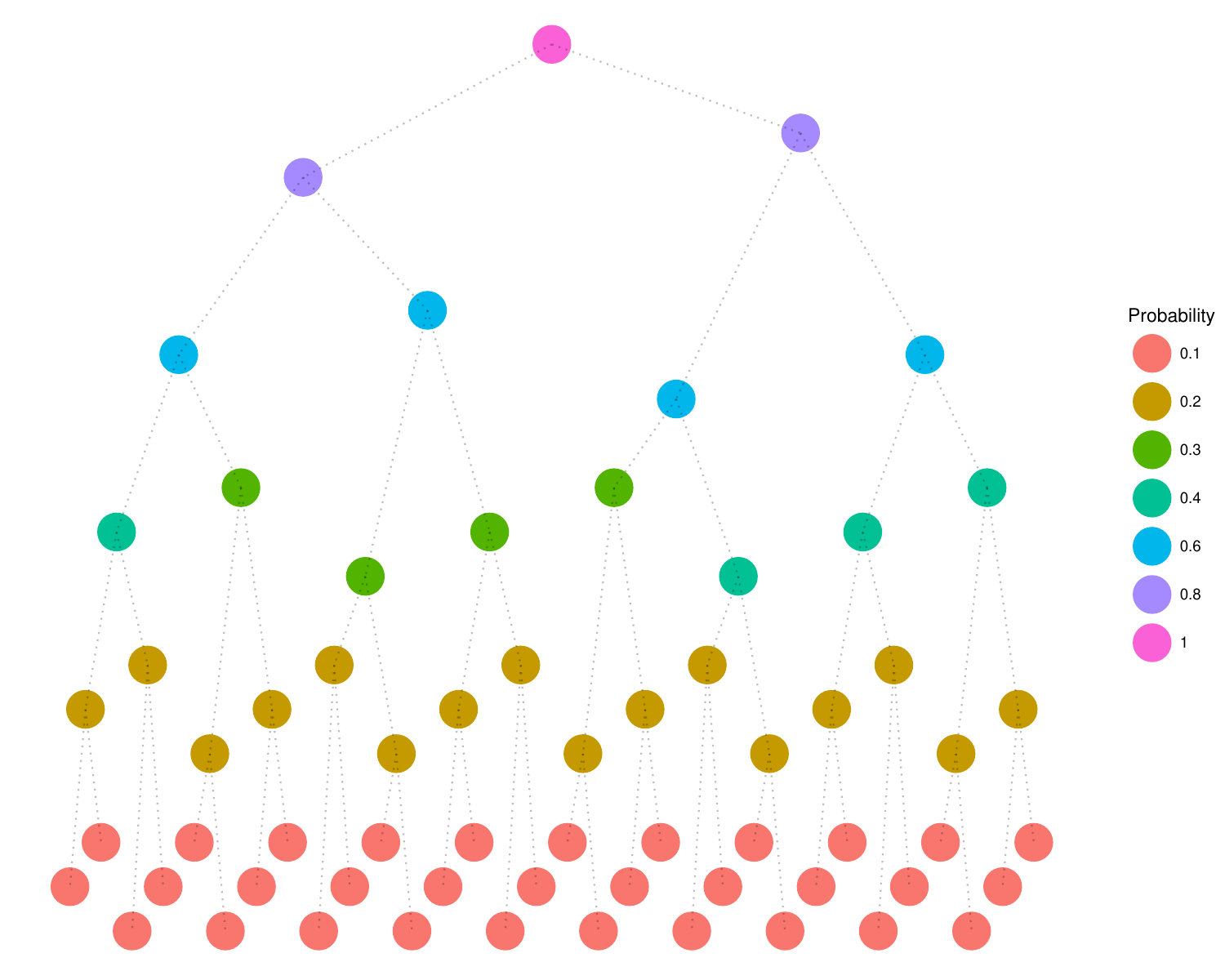}
  \caption{An example of the perfect binary tree, representing main effects. Node color represents the success probability (rounded to 1 decimal place) of the corresponding Bernoulli random variable.}
  \label{fig:btree}
\end{figure}
We can control the degree to which the interaction signal can be explained by main effects by choosing the proportion of nonzero elements of $\gamma^\ast$ correspond to interactions between main effects that are ancestors/descendants of each other versus not.
We consider three scenarios: (a) almost all interactions can be explained by main effects; (b) approximately half of the interactions can be explained by main effects; and (c) a very limited amount of interactions can be explained by main effects. These three scenario correspond to three cases where the {\em main-effect-interaction-ratio},
\begin{align}
  \mathrm{MIR} = \frac{\norm{\X \theta^\ast}_2^2}{\norm{\W \gamma^\ast}_2^2},
  \nonumber
\end{align}
is large, medium, and small. We would expect sprinter to perform especially well when $\mathrm{MIR}$ is large. For each value of MIR, we generate the response $\y$ using \eqref{eq:true_repa} with the zero-mean additive noise $\varepsilon$ generated according to the signal-to-noise ratio $\frac{\norm{\X \theta^\ast}^2 + \norm{\W \gamma^\ast}^2}{n\sigma^2} \in \left\{ 0.3, 0.5, 1, 1.5, 2, 2.5, 3\right\}$.
We generate $n = 100$ samples in each simulation setting, and in Figure \ref{fig:binary} we report the prediction error (on another $n = 100$ testing samples) of various methods (averaged over $50$ repetitions).
In particular, we compare the performance of the following methods:
\begin{itemize}
  \item The all pairs lasso (\texttt{APL}) with tuning parameter selected by cross-validation. We use the \texttt{R} package \texttt{glmnet} to implement APL.
  \item \texttt{sprinter}, as in Algorithm \ref{alg:first}, with lasso using main effects and squared effects in Step 1. In Step 2, we use the ``top-$m$ approach'' as in \eqref{eq:Ihat_m} with $m = \lceil n / \log(n) \rceil$. We use a two-dimensional cross-validation to select $(\lambda_1, \lambda_3)$, the tuning parameter pair for the lasso in Step 1 and Step 3, respectively. Both lasso fits are implemented using \texttt{glmnet}.
  \item The main effects lasso (\texttt{MEL}) with tuning parameter selected by cross-validation.
  \item \texttt{SIS and Lasso}: We use sure independence screening \citep{fan2008sure} on all main effects and interactions, and fit the lasso on the selected candidate features.
  \item Interaction Pursuit (\texttt{IP}) by \citet{2016arXiv160508933F}.
\end{itemize}
\begin{figure}[ht]
  \centering
  \captionsetup{width=.8\textwidth}
  \includegraphics[width = \textwidth]{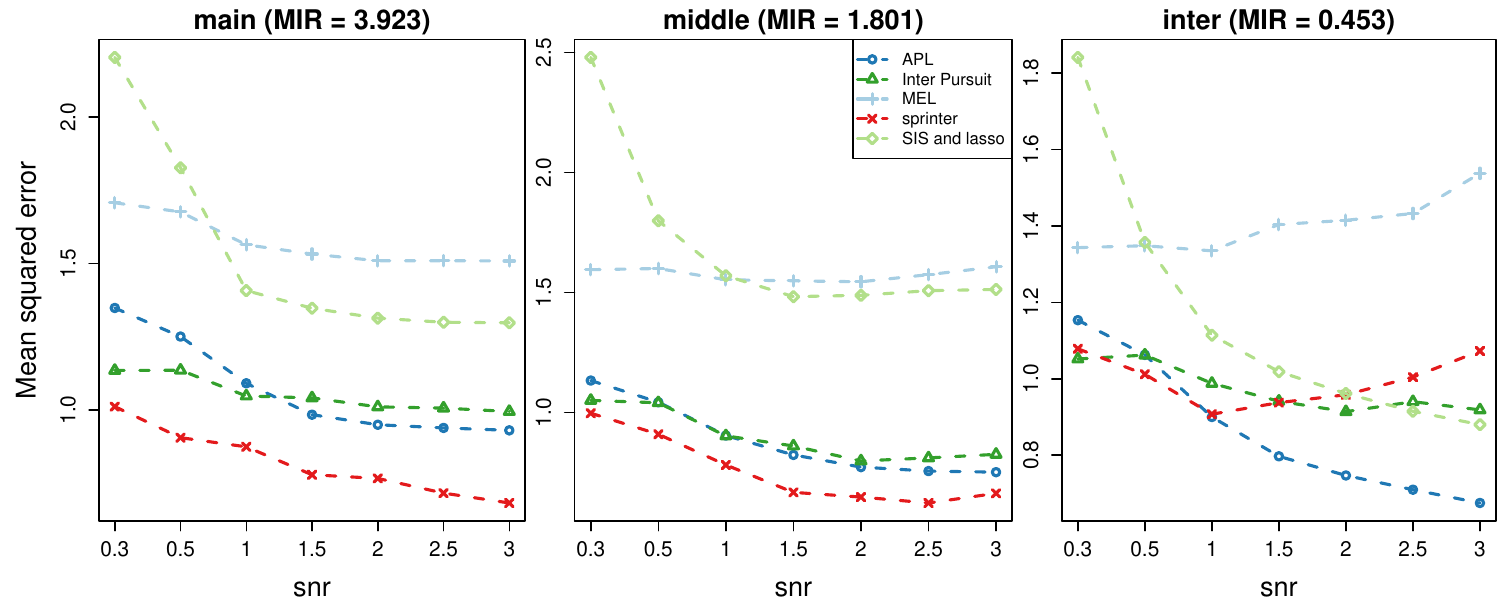}
  \caption{Prediction mean-squared error of different methods (averaged over 50 repetitions, binary settings).}
  \label{fig:binary}
\end{figure}
As MIR gets small, the performance of MEL worsens relative to other methods that model interactions. The performance of sprinter is favorable in comparison with other methods, especially when most of the signal can be well captured by main effects (as expected by the design of $\mathrm{MIR}$). The third panel shows that APL achieves favorable performance compared with other methods when the signal is sufficiently strong and concentrated primarily in pure interactions. As we will see in subsequent simulation studies, the performance advantage of sprinter over APL is more pronounced when $p$ is much larger.

Furthermore, Figure~\ref{fig:binary_nnzi} shows that the relationship between the number of selected interactions and the prediction error in the above three settings where the signal-to-noise ratio is 2. We observe that compared with IP, sprinter attains smaller prediction error using similar numbers of selected interactions. 
Compared with APL, sprinter selected a leaner set of interactions while achieving favorable prediction error (except in the last setting) .

\begin{figure}[ht]
  \centering
  \captionsetup{width=.8\textwidth}
  \includegraphics[width = \textwidth]{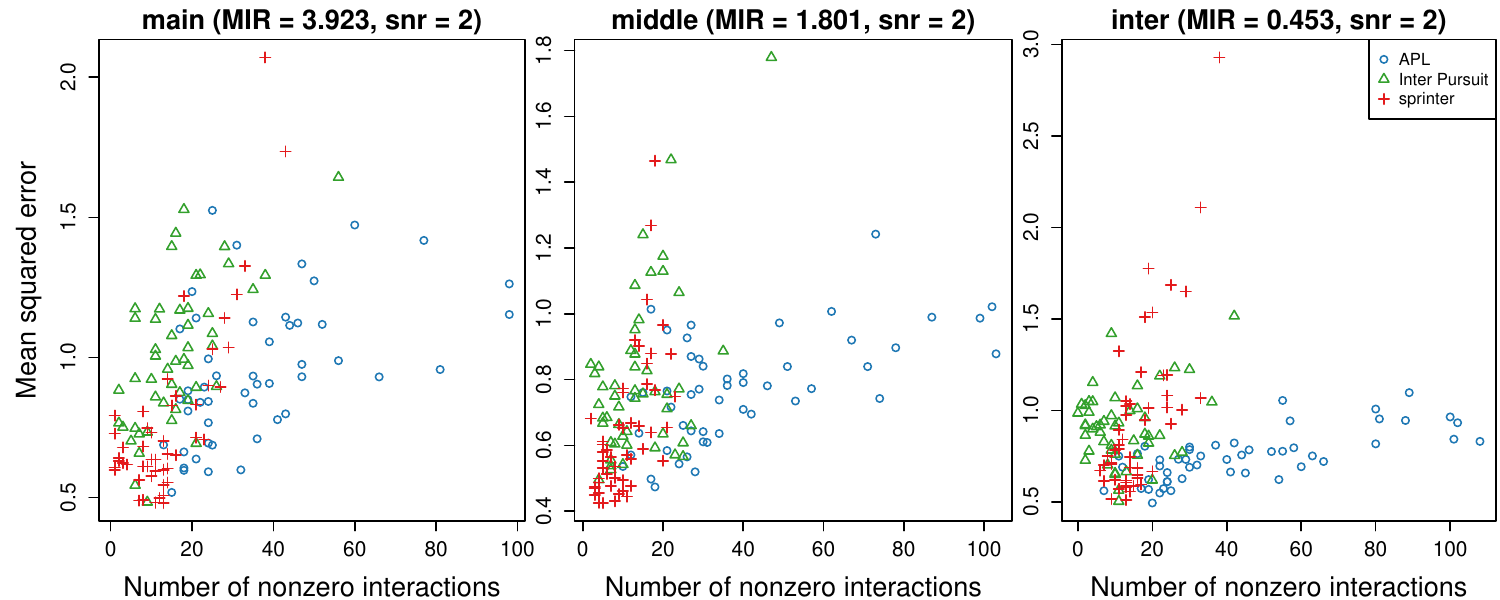}
  \caption{Prediction mean-squared error v.s. the number of selected interactions. Each point represents one (out of 50) repetition.}
  \label{fig:binary_nnzi}
\end{figure}

\subsection{Simulation studies: Gaussian features and hierarchy} \label{subsec:sim_Gaussian}
Next we study the performance of sprinter under different interaction structures when the main effects follow a multivariate Gaussian distribution.
We generate $n = 100$ samples from model \eqref{model:true}, where $X$ is a $p$-dimensional random vector following a multivariate zero-mean Gaussian distribution with $\Cov(X_j, X_k) = 0.5^{|j-k|}$ for $1 \leq j, k \leq p$, and $p = 400$.
Recall that our proposed idea of reluctance to interactions in Section \ref{sec:principle} is different from the hierarchical principle for interactions, and the proposed method does not assume hierarchy; actually sprinter does not assume any structure among interactions. 
Denote $\T_1$ as the indices of non-zero main effects, $\T_2$ as the indices of non-zero squared terms
and $\T_3$ as indices of non-zero interaction terms, and consider the following structures for the interactions:
\begin{enumerate}
  \item Mixed: $\T_1 = \left\{ 1, 2, \dots, 6 \right\}, \T_2 = \left\{ 1, 5, 15 \right\}, 
    \T_3 = \left\{ \left( 1, 5 \right), \left( 4, 18 \right), \left( 10, 11 \right), \left( 9, 17 \right), \left( 1, 13 \right), \left( 4, 17 \right) \right\}$.
  \item Hierarchical, i.e., $\beta_{jk} \neq 0 \implies \beta_j \neq 0 \text{ or } \beta_k \neq 0$:
    $\T_1 = \left\{ 1, 2, \dots, 6 \right\}$ , $\T_2 = \left\{ 1, 2, 3 \right\}$ and 
    $\T_3 = \left\{ (1, 3), (2, 4), \left( 3, 4 \right), (1, 8), (2,8), (5, 10) \right\}$.
  \item Anti-hierarchical, i.e., $\beta_{jk} \neq 0 \implies \beta_j = 0, \beta_k = 0$:
    $\T_1 = \left\{ 1, 2, \dots, 6 \right\}$, $\T_2 = \left\{ 11, 12, 13 \right\}$ and
    $\T_3 = \left\{ (11, 13), (12, 14), \left( 13, 14 \right), (11, 18), (12,18), (15, 20) \right\}$.
  \item Interaction only: $\T_1 = \T_2 = \emptyset$ and 
    $\T_3 = \left\{ (1, 3), (2, 4), \left( 3, 4 \right), (1, 8), (2,8), (5, 10) \right\}$.
  \item Main effects only:  
    $\T_1 = \left\{ 1, 2, \dots, 6 \right\}$ , $\T_2 = \emptyset$ and $\T_3 = \emptyset$.
  \item Squared effects only:  
    $\T_1 = \emptyset$ , $\T_2 = \left\{ 1, 2, \dots, 6 \right\}$ and $\T_3 = \emptyset$.
\end{enumerate}
Note that the hierarchy structure only exists in the hierarchical model and the main effects only model.
The signal strength is then set as $\beta^\ast_j = 2$ for $j \in \T_1$, $\gamma^\ast_j = 3$ for $j \in \T_2$ and $j \in \T_3$.
Finally, the zero-mean additive noise $\varepsilon$ in \eqref{model:true} is generated according to the signal-to-noise ratio $\sqrt{\frac{\norm{\X \beta^\ast}^2 + \norm{\Z \gamma^\ast}^2}{n\sigma^2}} \in \left\{ 0.3, 0.5, 1, 1.5, 2, 2.5, 3\right\}$.

In addition to the methods considered in the previous study, we also include the performance of the following two methods:
\begin{itemize}
  \item \texttt{sprinter(1cv)}, with everything the same as sprinter, except that cross-validation of $\lambda_1$ is performed before subsequent steps. See Section \ref{subsec:sprinter1cv} for details.
  \item \texttt{RAMP} \citep{hao2018model}, which iteratively adds variables into a path of solutions under the marginality (hierarchy) principle. 
    They also consider the two-stage lasso, but state that \texttt{RAMP} performs better than the two-stage lasso \citep{hao2014interaction}. The performance of \texttt{RAMP} is very unstable for binary features, and thus is not included in Section \ref{subsec:binary} or later in Section \ref{subsec:tripadvisor}.
\end{itemize}
\begin{figure}[ht]
  \centering
  \captionsetup{width=.8\textwidth}
  \includegraphics[width = .9\textwidth]{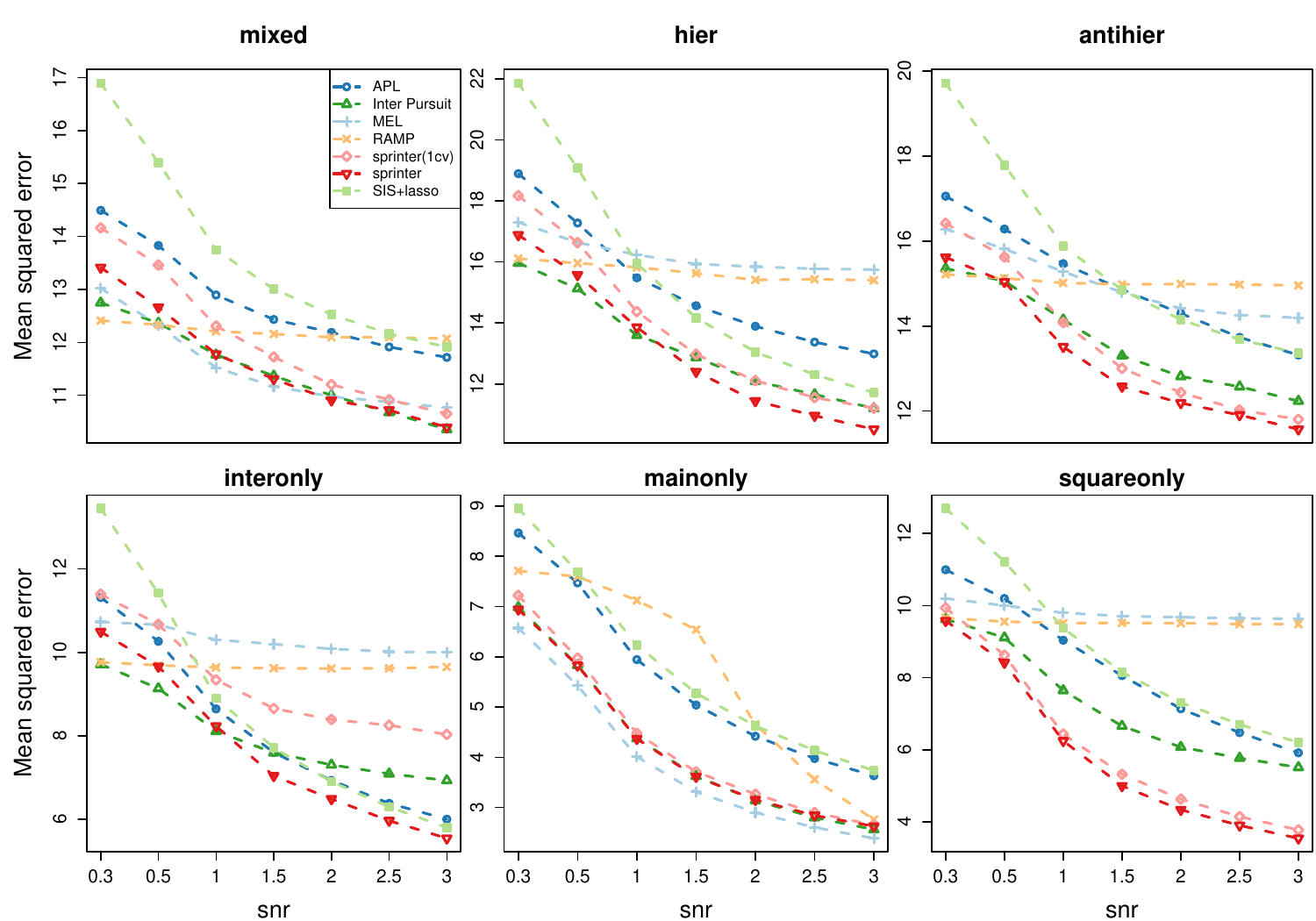}
  \caption{Prediction mean-squared error of different methods (averaged over 100 repetitions, Gaussian settings).}
  \label{fig:across_pred}
\end{figure}
We measure the statistical performance of each method in prediction error, which is averaged over 100 repetitions and is reported in Figure~\ref{fig:across_pred}.
Observe that sprinter almost works uniformly better than other methods in all settings except in the model where main effects are dominating. 
This is because sprinter includes both main effects and the squared effects in Step 1, which involves $p$ irrelevant squared effects. Actually, although not shown here, sprinter works much better in this setting if it uses only main effects in Step 1.

We also note that the prediction performance advantage of \texttt{sprinter} over \texttt{sprinter(1cv)} is rather marginal, except for the interaction-only model. As discussed in Section \ref{subsec:sprinter1cv}, when there are only interactions in the model, \texttt{sprinter(1cv)} tends to under-penalize the main effects in the cross-validation of Step 1, and thus has much worse performance than \texttt{sprinter}.

\subsection{Simulation studies: computation time}
In this section, we show that sprinter is much more computationally efficient than APL, while having similar (if not better) statistical performance. To this end, we consider varying  $p \in \left\{ 100, 200, 1000, 2000 \right\}$ in the mixed model in Section \ref{subsec:sim_Gaussian} with signal-to-noise ratio equal to $3$. The following plots show both the computation time (in seconds) and the prediction mean squared error (averaged over 100 repetitions).
\begin{figure}[ht]
  \centering
  \captionsetup{width=.8\textwidth}
  \includegraphics[width = .9\textwidth]{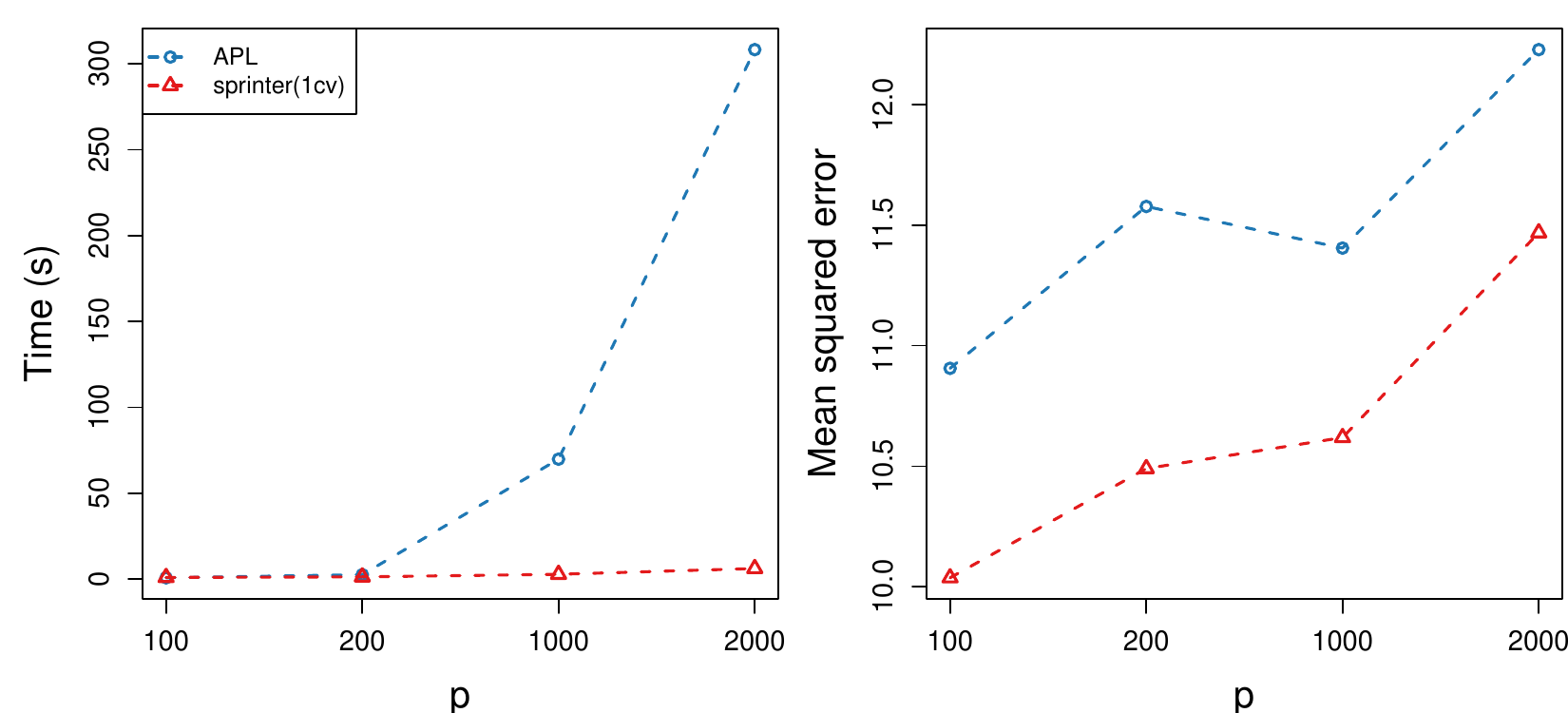}
  \caption{Computation time and prediction mean-squared error (averaged over 100 repetitions) for different $p$ in the mixed model.}
  \label{fig:pred_time}
\end{figure}
As expected, APL is computationally much more expensive than the proposed method. In particular, for $p = 2000$, the proposed method is about 100 times faster than APL. In addition, while not shown, sprinter can solve a problem with $140000$ main effects (about 10 billion interactions, which is infeasible for APL) with 5-fold cross-validation under 7 hours on a single CPU. 

In addition to enjoying obvious computational benefits, the right panel of Figure \ref{fig:pred_time} shows that the proposed method does not lose statistical performance in terms of prediction error. Actually, sprinter attains uniformly smaller prediction error than APL. 

Furthermore, Figure~\ref{fig:nnzm_nnzi} shows that compared with APL, sprinter achieves this favorable performance with a much smaller number of selected interactions (and more selected main effects) --- a property that is expected by prioritizing main effects over interactions, and is beneficial for interpretation. 
\begin{figure}[ht]
  \centering
  \captionsetup{width=.8\textwidth}
  \includegraphics[width = \textwidth]{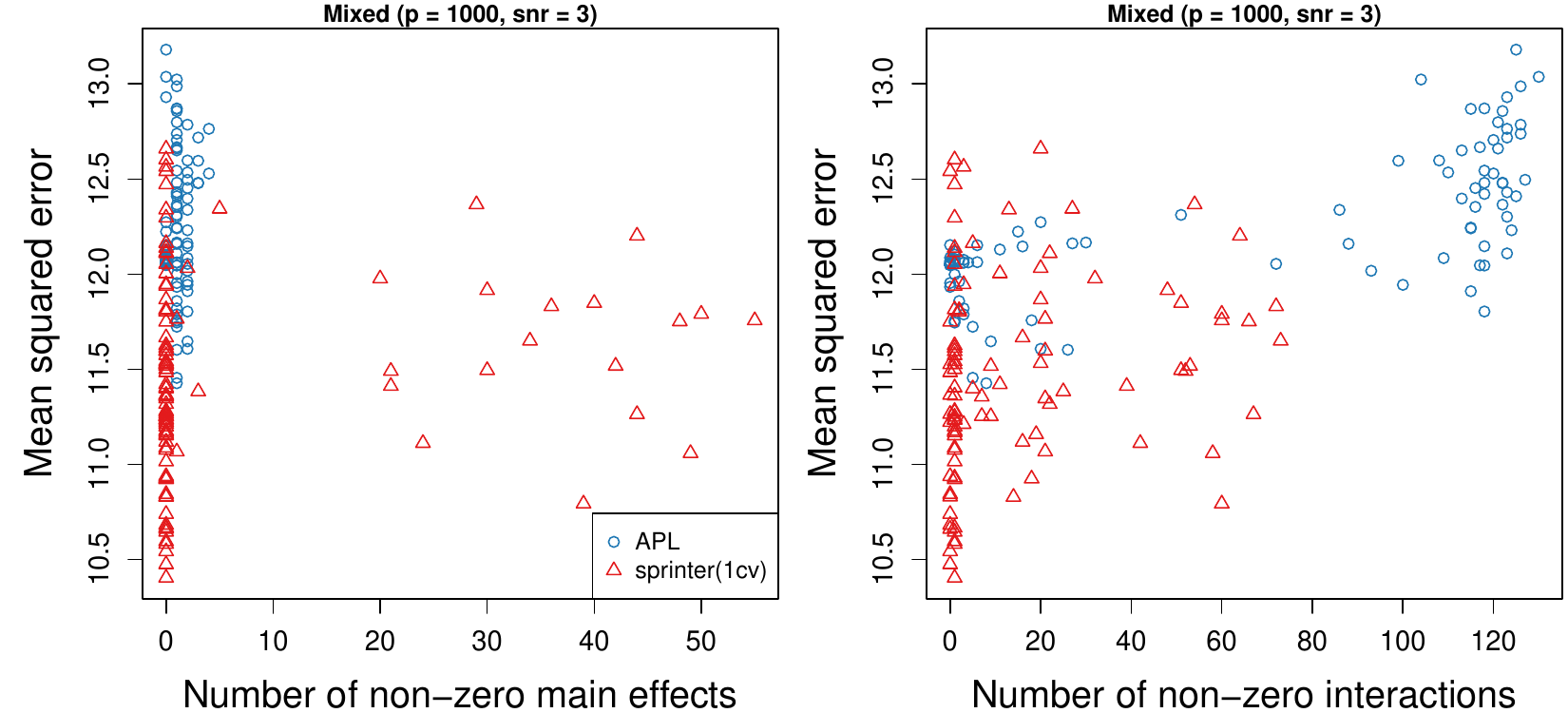}
  \caption{Number of non-zero main effects, interactions, and prediction mean-squared error for $p = 1000$ in the mixed model. Each point represents one (out of 100) repetition.}
  \label{fig:nnzm_nnzi}
\end{figure}


\subsection{Data example: Tripadvisor hotel reviews} \label{subsec:tripadvisor}
We consider the dataset prepared in \citet{wang2010latent}, which contains reviews and the corresponding ratings (on the integer scale of $\{1, \ldots, 5\}$) of hotels on \texttt{Tripadvisor.com}. 
We further construct a dictionary, consisting of $p = 7817$ distinct words following the approach in \citet{yan2020rare}. Most of the words are adjectives, while we also include words that might alter the meaning of an adjective such as \textit{but}, \textit{not}, \textit{should}, etc.  
In total, the dataset contains $n = 211321$ reviews and is represented in the binary design matrix $\X \in \real^{n \times p}$, where $\X_{ij} = 1$ if the $j$-th word in the dictionary appears in the $i$-th review, and $\X_{ij} = 0$ otherwise.

Our goal is to predict the hotel rating based on the words used in the review. While \eqref{model:true} is certainly an overly simplistic model for text data (especially in light of advances in natural language processing in recent years), the ability to easily interpret words makes this a helpful illustration of our method. Note that  $\X_{ij} \X_{ik}=1$ indicates that both the $j$-th word and the $k$-th word appear in the $i$-th review. The co-occurrence of two key words can provide useful information in the prediction task.
For example, while a negative sounding word like \textit{``worst''} might be associated with lower ratings, its co-occurrence with words like \textit{``not''} (as in {\it ``not the worst''}) would yield a very important interaction for prediction. 

Due to memory constraints (that applies to all methods considered), we randomly select $10\%$ of the whole data as a training set, and the remaining $90\%$ as a testing set.  We compare several methods: \texttt{MEL}, a two-stage hierarchical lasso \citep[our implementation of a method proposed in][]{hao2018model}, and \texttt{sprinter}. The problem size is too big for APL.
In Table \ref{tab:tripadvisor}, we report the prediction root mean squared error on the testing set. In addition, we report the number of selected main effects and interactions. The performance of the other methods in previous sections is much worse, and thus are not included. 
\begin{table}[H]
  \centering
  \begin{tabular}{ |l | c | c | c |}
    \hline
    method & prediction RMSE & \# main effects & \# interactions \\ \hline
    MEL & 1.153 & 689 & 0  \\ \hline
    Two Stage Hierarchical Lasso & 1.125 & 500 & 992 \\ \hline
    sprinter & 1.091 & 304 & 526 \\ \hline
  \end{tabular}
  \captionsetup{width=.85\textwidth}
  \caption{Tripadvisor hotel reviews: prediction root mean squared error on the testing set, number of selected main effects and interactions}
  \label{tab:tripadvisor}
\end{table}
Comparing with sprinter, the prediction error of MEL suggests the importance of including interactions in the model, and the prediction error of the two-stage hierarchical lasso may suggest that the true underlying interaction structure is not hierarchical. 

In the sprinter model, the top $5$ main effects with largest positive coefficient estimates are \textit{excellent}, \textit{fantastic}, \textit{perfect}, \textit{wonderful}, and \textit{loved}, while the top $5$ main effects with largest negative coefficient estimates are \textit{worst}, \textit{rude}, \textit{poor}, \textit{dirty}, and \textit{awful}. The sign and the magnitude of the coefficient estimates corresponding to these words are of course highly aligned with the words' actual meanings.

The coefficient estimates for interactions are quite interesting. As expected, the interactions with the large positive coefficients include \textit{awful}$\times$\textit{not} and \textit{rude}$\times$\textit{but}.  This makes sense since the words {\em not} and {\em but} often have the effect of negating or at least weakening the sentiment of the words around them.   Interestingly, the coefficient corresponding to the interaction \textit{wonderful}$\times$\textit{superb}  and \textit{great} $\times$ \textit{fantastic} have large negative estimate.  We view this as reflecting the diminishing returns of superlative synonyms: a review containing both \textit{wonderful} and \textit{superb} does not make the rating twice as good as if the review only contains one of the words.  The interaction serves to reduce the additive effects of the use of synonyms. This is especially sensible given that we are using presence-absence coding of words because it means that using \textit{wonderful} and \textit{superb} together is treated the same as if someone used \textit{wonderful} twice.

\section{Discussion}
In this paper, we present sprinter (sparse reluctant interaction modeling),  a general multi-step framework that is highly scalable to large interaction problems. Our proposal is motivated by a new guiding principle in interaction modeling, which favors main effects over interactions if all else is equal. We show, in both theoretical analysis and numerical studies, that sprinter achieves strong computational efficiency, without compromising favorable statistical properties.

Although not discussed in this paper, the sprinter framework can be easily generalized to higher order interaction modeling. The idea of screening interactions based on the residual from fitting a model of lower order terms can be used repeatedly to select higher order interactions. For example, after Step 3 of Algorithm \ref{alg:first}, we could obtain a residual vector, which could then be used in screening three-way interactions. This approach can be used recursively in modeling higher order interactions, by prioritizing lower-order terms over higher-order terms if all else is equal.

The idea of enforcing reluctance to interactions could be further generalized to a broader context beyond interaction modeling.
This is particularly useful when side information makes us prefer one group of variables over another. For example, one might be reluctant in using a group of expensive variables, and thus would try to capture the response as much as possible by first only using the group of cheap variables. This is an interesting future research direction.

One can also extend the proposed reluctant interaction modeling framework to using a general loss function $\ell$ (e.g., as in generalized linear models), in which case the screening in Step 2 could be based on absolute correlation between each interaction and $\nabla \ell$ evaluated in Step 1.

Finally, an R \citep{citeR} package, named \texttt{sprintr}, is available online, implementing our method. The estimation is very fast with the core screening functions coded in \texttt{Rcpp} \citep{rcpp}.
The simulation studies in Section \ref{sec:num} use the \texttt{simulator} package \citep{2016arXiv160700021B}, and the code to reproduce all numerical studies is available at \url{https://github.com/hugogogo/reproducible/tree/master/sprinter}. The \texttt{Tripadvisor} data are available on request from the corresponding author. These data were derived from the following resources available in the public domain: \url{http://times.cs.uiuc.edu/~wang296/Data/}.

\section*{Acknowledgment}
We thank Daniela Witten for helpful feedback. 
GY and JB were supported by NIH Grant R01GM123993. JB was also supported by NSF CAREER Award DMS-1653017. RT was supported by NSF CAREER Award DMS-1554123.

\vskip 0.2in
\bibliographystyle{agsm}
\bibliography{short.bib}

@inproceedings{shamos1975closest,
  title={Closest-point problems},
  author={Shamos, Michael Ian and Hoey, Dan},
  booktitle={16th Annual Symposium on Foundations of Computer Science (sfcs 1975)},
  pages={151--162},
  year={1975},
  organization={IEEE}
}

@article{agarwal1991euclidean,
  title={Euclidean minimum spanning trees and bichromatic closest pairs},
  author={Agarwal, Pankaj K and Edelsbrunner, Herbert and Schwarzkopf, Otfried and Welzl, Emo},
  journal={Discrete \& Computational Geometry},
  volume={6},
  number={3},
  pages={407--422},
  year={1991},
  publisher={Springer}
}

@book{rajaraman2011mining,
  title={Mining of massive datasets},
  author={Rajaraman, Anand and Ullman, Jeffrey David},
  year={2011},
  publisher={Cambridge University Press}
}

@article{2016arXiv160700021B,
	Archiveprefix = {arXiv},
	Author = {{Bien}, J.},
	Eprint = {1607.00021},
	Journal = {ArXiv e-prints},
	Month = jun,
	Title = {{The Simulator: An Engine to Streamline Simulations}},
	Year = 2016}

@Article{rcpp,
  title = {{Rcpp}: Seamless {R} and {C++} Integration},
  author = {Dirk Eddelbuettel and Romain Fran\c{c}ois},
  journal = {Journal of Statistical Software},
  year = {2011},
  volume = {40},
  number = {8},
  pages = {1--18},
  url = {https://www.jstatsoft.org/v40/i08/},
  doi = {10.18637/jss.v040.i08},
}

@article{friedman2010regularization,
	author = {Friedman, Jerome and Hastie, Trevor and Tibshirani, Rob},
	journal = {Journal of statistical software},
	number = {1},
	pages = {1},
	publisher = {NIH Public Access},
	title = {Regularization paths for generalized linear models via coordinate descent},
	volume = {33},
	year = {2010}}

@article{buja2019models,
	author = {Buja, Andreas and Brown, Lawrence and Berk, Richard and George, Edward and Pitkin, Emil and Traskin, Mikhail and Zhang, Kai and Zhao, Linda and others},
	journal = {Statistical Science},
	number = {4},
	pages = {523--544},
	publisher = {Institute of Mathematical Statistics},
	title = {Models as approximations I: Consequences illustrated with linear regression},
	volume = {34},
	year = {2019}}

@article{wang2019penalized,
  title={Penalized interaction estimation for ultrahigh dimensional quadratic regression},
  author={Wang, Cheng and Jiang, Binyan and Zhu, Liping},
  journal={Statistica Sinica},
  volume={31},
  number={3},
  pages={1549--1570},
  year={2021},
  publisher={JSTOR}
}

@article{zhou2019bolt,
  title={BOLT-SSI: A statistical approach to screening interaction effects for ultra-High dimensional data},
  author={Zhou, Min and Dai, Mingwei and Yao, Yuan and Liu, Jin and Yang, Can and Peng, Heng},
  journal={arXiv preprint arXiv:1902.03525},
  year={2019}
}

@article{yan2020rare,
	author = {Yan, Xiaohan and Bien, Jacob},
	journal = {Journal of the American Statistical Association},
	pages = {1--14},
	publisher = {Taylor \& Francis},
	title = {Rare feature selection in high dimensions},
	year = {2020}}

@inproceedings{wang2010latent,
  title={Latent aspect rating analysis on review text data: a rating regression approach},
  author={Wang, Hongning and Lu, Yue and Zhai, Chengxiang},
  booktitle={Proceedings of the 16th ACM SIGKDD international conference on Knowledge discovery and data mining},
  pages={783--792},
  year={2010}
}

@manual{citeR,
	Address = {Vienna, Austria},
	Author = {{R Core Team}},
	Organization = {R Foundation for Statistical Computing},
	Title = {R: A Language and Environment for Statistical Computing},
	Url = {https://www.R-project.org/},
	Year = {2018}}

@article{2018arXiv180107785R,
	Adsurl = {http://adsabs.harvard.edu/abs/2018arXiv180107785R},
	Archiveprefix = {arXiv},
	Author = {{Reese}, R. and {Dai}, X. and {Fu}, G.},
	Eprint = {1801.07785},
	Journal = {ArXiv e-prints},
	Month = jan,
	Primaryclass = {stat.ME},
	Title = {{Strong Sure Screening of Ultra-high Dimensional Data with Interaction Effects}},
	Year = 2018}

@inproceedings{hazimeh2020learning,
	author = {Hazimeh, Hussein and Mazumder, Rahul},
	booktitle = {International Conference on Artificial Intelligence and Statistics},
	organization = {PMLR},
	pages = {1833--1843},
	title = {Learning hierarchical interactions at scale: A convex optimization approach},
	year = {2020}}

@article{kuchibhotla2018moving,
	title={Moving beyond sub-Gaussianity in high-dimensional statistics: Applications in covariance estimation and linear regression},
  author={Kuchibhotla, Arun Kumar and Chakrabortty, Abhishek},
  journal={Information and Inference: A Journal of the IMA},
  volume={11},
  number={4},
  pages={1389--1456},
  year={2022},
  publisher={Oxford University Press}}

@article{hao2018model,
	Author = {Hao, Ning and Feng, Yang and Zhang, Hao Helen},
	Journal = {Journal of the American Statistical Association},
	Number = {522},
	Pages = {615--625},
	Publisher = {Taylor \& Francis},
	Title = {Model selection for high-dimensional quadratic regression via regularization},
	Volume = {113},
	Year = {2018}}

@article{hao2017note,
	Author = {Hao, Ning and Zhang, Hao Helen},
	Journal = {The American Statistician},
	Number = {4},
	Pages = {291--297},
	Publisher = {Taylor \& Francis},
	Title = {A note on high-dimensional linear regression with interactions},
	Volume = {71},
	Year = {2017}}

@article{wainwright2009sharp,
	Author = {Wainwright, Martin J},
	Journal = {Information Theory, IEEE Transactions on},
	Number = {5},
	Pages = {2183--2202},
	Publisher = {IEEE},
	Title = {Sharp thresholds for high-dimensional and noisy sparsity recovery using-constrained quadratic programming (Lasso)},
	Volume = {55},
	Year = {2009}}

@article{tibshirani1996regression,
  author =        {Tibshirani, Robert},
  journal =       {Journal of the Royal Statistical Society. Series B
                   (Methodological)},
  number =        {1},
  pages =         {267--288},
  publisher =     {JSTOR},
  title =         {Regression shrinkage and selection via the lasso},
  volume =        {58},
  year =          {1996},
}

@article{fan2008sure,
  author =        {Fan, Jianqing and Lv, Jinchi},
  journal =       {Journal of the Royal Statistical Society: Series B
                   (Statistical Methodology)},
  number =        {5},
  pages =         {849--911},
  publisher =     {Wiley Online Library},
  title =         {Sure independence screening for ultrahigh dimensional
                   feature space},
  volume =        {70},
  year =          {2008},
}

@article{peixoto1987hierarchical,
  author =        {Peixoto, Julio L},
  journal =       {The American Statistician},
  number =        {4},
  pages =         {311--313},
  publisher =     {Taylor \& Francis Group},
  title =         {Hierarchical variable selection in polynomial
                   regression models},
  volume =        {41},
  year =          {1987},
}

@article{hamada1992analysis,
  author =        {Hamada, Michael and Wu, CF Jeff},
  journal =       {Journal of quality technology},
  number =        {3},
  pages =         {130--137},
  publisher =     {Taylor \& Francis},
  title =         {Analysis of designed experiments with complex
                   aliasing},
  volume =        {24},
  year =          {1992},
}

@article{nelder1977reformulation,
  author =        {Nelder, JA},
  journal =       {Journal of the Royal Statistical Society. Series A
                   (General)},
  pages =         {48--77},
  publisher =     {JSTOR},
  title =         {A reformulation of linear models},
  year =          {1977},
}

@article{efron2004least,
  author =        {Efron, Bradley and Hastie, Trevor and Johnstone, Iain and
                   Tibshirani, Robert and others},
  journal =       {The Annals of statistics},
  number =        {2},
  pages =         {407--499},
  publisher =     {Institute of Mathematical Statistics},
  title =         {Least angle regression},
  volume =        {32},
  year =          {2004},
}

@article{turlach2004discussion,
	author = {Berwin Turlach},
	issn = {0090-5364},
	journal = {The Annals of Statistics},
	language = {English},
	number = {2},
	pages = {481--490},
	publisher = {Institute of Mathematical Statistics},
	title = {Discussion of Least Angle Regression by Efron et al},
	volume = {32},
	year = {2004}}

@article{zhao2009composite,
  author =        {Zhao, Peng and Rocha, Guilherme and Yu, Bin},
  journal =       {The Annals of Statistics},
  number =        {6A},
  pages =         {3468--3497},
  publisher =     {JSTOR},
  title =         {The composite absolute penalties family for grouped
                   and hierarchical variable selection},
  volume =        {37},
  year =          {2009},
}

@article{yuan2009structured,
  author =        {Yuan, Ming and Joseph, V Roshan and Zou, Hui},
  journal =       {The Annals of Applied Statistics},
  pages =         {1738--1757},
  publisher =     {JSTOR},
  title =         {Structured variable selection and estimation},
  year =          {2009},
}

@article{choi2010variable,
  author =        {Choi, Nam Hee and Li, William and Zhu, Ji},
  journal =       {Journal of the American Statistical Association},
  number =        {489},
  pages =         {354--364},
  publisher =     {Taylor \& Francis},
  title =         {Variable selection with the strong heredity
                   constraint and its oracle property},
  volume =        {105},
  year =          {2010},
}

@article{radchenko2010variable,
  author =        {Radchenko, Peter and James, Gareth M},
  journal =       {Journal of the American Statistical Association},
  number =        {492},
  pages =         {1541--1553},
  publisher =     {Taylor \& Francis},
  title =         {Variable selection using adaptive nonlinear
                   interaction structures in high dimensions},
  volume =        {105},
  year =          {2010},
}

@inproceedings{schmidt2010convex,
  author =        {Schmidt, Mark and Murphy, Kevin},
  booktitle =     {Proceedings of the Thirteenth International
                   Conference on Artificial Intelligence and Statistics},
  pages =         {709--716},
  title =         {Convex structure learning in log-linear models:
                   Beyond pairwise potentials},
  year =          {2010},
}

@article{bien2013lasso,
  title={A lasso for hierarchical interactions},
  author={Bien, Jacob and Taylor, Jonathan and Tibshirani, Robert},
  journal={Annals of statistics},
  volume={41},
  number={3},
  pages={1111},
  year={2013},
  publisher={NIH Public Access}
}

@article{lim2015learning,
  author =        {Lim, Michael and Hastie, Trevor},
  journal =       {Journal of Computational and Graphical Statistics},
  number =        {3},
  pages =         {627--654},
  publisher =     {Taylor \& Francis},
  title =         {Learning interactions via hierarchical group-lasso
                   regularization},
  volume =        {24},
  year =          {2015},
}

@article{haris2016convex,
  author =        {Haris, Asad and Witten, Daniela and Simon, Noah},
  journal =       {Journal of Computational and Graphical Statistics},
  number =        {4},
  pages =         {981--1004},
  publisher =     {Taylor \& Francis},
  title =         {Convex modeling of interactions with strong heredity},
  volume =        {25},
  year =          {2016},
}

@article{she2018group,
  author =        {She, Yiyuan and Wang, Zhifeng and Jiang, He},
  journal =       {Journal of the American Statistical Association},
  number =        {521},
  pages =         {445--454},
  publisher =     {Taylor \& Francis},
  title =         {Group regularized estimation under structural
                   hierarchy},
  volume =        {113},
  year =          {2018},
}

@article{wu2009genome,
  author =        {Wu, Tong Tong and Chen, Yi Fang and Hastie, Trevor and
                   Sobel, Eric and Lange, Kenneth},
  journal =       {Bioinformatics},
  number =        {6},
  pages =         {714--721},
  publisher =     {Oxford University Press},
  title =         {Genome-wide association analysis by lasso penalized
                   logistic regression},
  volume =        {25},
  year =          {2009},
}

@article{wu2010screen,
  author =        {Wu, Jing and Devlin, Bernie and Ringquist, Steven and
                   Trucco, Massimo and Roeder, Kathryn},
  journal =       {Genetic epidemiology},
  number =        {3},
  pages =         {275--285},
  publisher =     {Wiley Online Library},
  title =         {Screen and clean: a tool for identifying interactions
                   in genome-wide association studies},
  volume =        {34},
  year =          {2010},
}

@article{hao2014interaction,
  author =        {Hao, Ning and Zhang, Hao Helen},
  journal =       {Journal of the American Statistical Association},
  number =        {507},
  pages =         {1285--1301},
  publisher =     {Taylor \& Francis},
  title =         {Interaction screening for ultrahigh-dimensional data},
  volume =        {109},
  year =          {2014},
}

@article{shah2016modelling,
  author =        {Shah, Rajen D},
  journal =       {Journal of Machine Learning Research},
  number =        {207},
  pages =         {1--31},
  title =         {Modelling interactions in high-dimensional data with
                   backtracking},
  volume =        {17},
  year =          {2016},
}

@article{culverhouse2002perspective,
  author =        {Culverhouse, Robert and Suarez, Brian K and
                   Lin, Jennifer and Reich, Theodore},
  journal =       {The American Journal of Human Genetics},
  number =        {2},
  pages =         {461--471},
  publisher =     {Elsevier},
  title =         {A perspective on epistasis: limits of models
                   displaying no main effect},
  volume =        {70},
  year =          {2002},
}

@article{2016arXiv160508933F,
  author =        {{Fan}, Y. and {Kong}, Y. and {Li}, D. and {Lv}, J.},
  journal =       {ArXiv e-prints},
  month =         may,
  title =         {{Interaction Pursuit with Feature Screening and
                   Selection}},
  year =          {2016},
}

@article{2016arXiv161005108T,
	Author = {Gian-Andrea Thanei and Nicolai Meinshausen and Rajen D. Shah},
	Journal = {Journal of Machine Learning Research},
	Number = {37},
	Pages = {1-42},
	Title = {The xyz algorithm for fast interaction search in high-dimensional data},
	Volume = {19},
	Year = {2018}}

@article{barut2015conditional,
	Author = {Barut, Emre and Fan, Jianqing and Verhasselt, Anneleen},
	Journal = {Journal of the American Statistical Association},
	Number = {515},
	Pages = {1266--1277},
	Publisher = {Taylor \& Francis},
	Title = {Conditional sure independence screening},
	Volume = {111},
	Year = {2016}}

@article{niu2018interaction,
  author =        {Niu, Yue Selena and Hao, Ning and Zhang, Hao Helen},
  journal =       {Statistics and Its Interface},
  number =        {2},
  pages =         {317--325},
  publisher =     {International Press of Boston},
  title =         {Interaction screening by partial correlation},
  volume =        {11},
  year =          {2018},
}

@book{cormen2009introduction,
  author =        {Cormen, Thomas H and Leiserson, Charles E and
                   Rivest, Ronald L and Stein, Clifford},
  publisher =     {MIT press},
  title =         {Introduction to algorithms},
  year =          {2009},
}

@article{vershynin2010introduction,
  author =        {Vershynin, Roman},
  journal =       {arXiv preprint arXiv:1011.3027},
  title =         {Introduction to the non-asymptotic analysis of random
                   matrices},
  year =          {2010},
}

\newpage

\appendices
\textbf{Organization:}
Appendix \ref{app:subweibull} first gives the technical tools for the theoretical analysis. Then Appendix \ref{app:proof_screening}, \ref{app:proof_prediction}, and \ref{app:proof_step1} prove the building blocks in order to prove Theorem \ref{thm:main}, which is then given in Appendix \ref{app:proof_thm_main}.
Appendix \ref{app:screening_m} provides the theoretical guarantees of Step 2 with the ``top-m'' approach.
In Appendix \ref{app:example} we give detailed discussions on the validity of the conditions such that Theorem \ref{thm:screening} holds in examples of Gaussian and Bernoulli main effects.

\section{Useful inequalities for sub-Weibull random variables} \label{app:subweibull}
We first present the following property of a $\subw$ random variable. 
\begin{lemma} \label{lem:exp}
  If $U \sim \subw(\nu)$ with norm $\snorm{U}_{\psi_\nu}$, then for any integer $k \geq 1$, we have
  \begin{align}
    \E \left[ |U|^k \right] \leq 2 \snorm{U}_{\psi_\nu}^k \frac{k}{\nu} \Gamma \left( \frac{k}{\nu}  \right),
  \end{align}
  where $\Gamma(x) = \int e^{-t} t^{x - 1} dt$ is the Gamma function.
\end{lemma}
\begin{proof}
  First we have
  \begin{align}
    \Prob \left( |U| > x \right) &= \Prob \left[ \exp\left( \frac{|U|^\nu}{\snorm{U}_{\psi_\nu}^\nu}  \right) > \exp \left( \frac{x^\nu}{\snorm{U}_{\psi_\nu}^\nu}  \right)  \right] \nonumber\\
                                 &\leq\E \left[ \exp\left( \frac{|U|^\nu}{\snorm{U}_{\psi_\nu}^\nu}  \right) \right]\exp \left( -\frac{x^\nu}{\snorm{U}_{\psi_\nu}^\nu}  \right) \leq 2 \exp \left( -\frac{x^\nu}{\snorm{U}_{\psi_\nu}^\nu}  \right),
                                 \nonumber
  \end{align}
  where the first inequality is Markov inequality, and the second inequality holds from Definition \ref{def:sb}.
  Then 
  \begin{align}
    \E \left[ |U|^k \right] &= \int_{0}^\infty \Prob \left( |U|^k > x \right) dx = \int_0^\infty \Prob \left( |U| > x^{1/k} \right) dx \nonumber\\
                            &\leq 2 \int_0^\infty \exp \left[ - \frac{x^{\nu / k}}{\snorm{U}_{\psi_\nu}^\nu}  \right] dx
                            = 2 \snorm{U}_{\psi_\nu}^k \frac{k}{\nu} \int_{0}^\infty \exp(-t) t^{\frac{k}{\nu} - 1} dt = 2 \snorm{U}_{\psi_\nu}^k \frac{k}{\nu} \Gamma\left( \frac{k}{\nu} \right),
                            \nonumber
  \end{align}
  where we use the change of variable $t = \frac{x^{\nu / k}}{\snorm{U}_{\psi_\nu}^\nu} $.
\end{proof}

The following theorem serves as the main tool for our theoretical analysis. It gives concentration inequalities for the average of $n$ $\iid$ $\subw(\nu)$ random variables. As the definition of $\subw(\nu)$ a generalization of sub-Gaussian and sub-Exponential random variables, the following theorem reduces to Hoeffding's inequality for sub-Gaussian random variables when $\nu = 2$, and it reduces to Bernstein inequality for sub-Exponential random variables when $\nu = 1$.
\begin{theorem} \label{thm:GBO}
  If $U_1, ..., U_n$ are $\iid$ $\subw(\nu)$ random variables with $\nu \leq 1$, then the following bound holds:
  \begin{align}
    \Prob \left( \left| \frac{1}{n} \sum_{i = 1}^n U_i - \E [U_1] \right| \geq C(\nu) \norm{U_1}_{\psi_\nu} \frac{t^{ \frac{2 + \nu}{4\nu} }}{n^{3/4}} \right)   \leq 2 e^{-t},
    \label{eq:concentration}
  \end{align} 
  where $C(\nu) > 0$ only depends on $\nu$.
\end{theorem}
\begin{proof}
  We first consider the following definition \citep[Definition 2.3]{kuchibhotla2018moving}:
  \begin{definition}[Generalized Berstein-Orlicz norm] \label{def:GBO}
    For fixed value of $\nu > 0$ and $L > 0$, define the function $\Psi_{\nu, L}$ based on its inverse function, for all $t \geq 0$,
    \begin{align}
      \Psi^{-1}_{\nu, L}(t) = \sqrt{\log (1 + t)} + L \left( \log (1 + t) \right)^{1 / \nu}.
      \nonumber
    \end{align}
    Then the generalized Berstein-Orlicz (GBO) norm of a random variable $U$ is defined as
    \begin{align}
      \norm{U}_{\Psi_{\nu, L}} = \inf \left\{ \zeta > 0: \E \left[ \Psi_{\nu, L}\left(\frac{|U|}{\zeta}\right) \right] \leq 1 \right\}.
      \nonumber
    \end{align}
  \end{definition}

  First it is easy to verify that $\Psi_{\nu, L}$ is monotonically non-decreasing, $\Psi_{\nu, L}(0) = 0$, and $\Psi_{\nu, L}(a) \geq 0$ for all $a \geq 0$.
  Then for $U \sim \subw(\nu)$,
  \begin{align}
    \Prob \left[ |U| \geq \snorm{U}_{\Psi_{\nu, L}} \left( \sqrt{t} + L t^{1 / \nu} \right) \right] &= \Prob \left[ \frac{|U|}{\snorm{U}_{\Psi_{\nu, L}}} \geq  \sqrt{t} + L t^{1 / \nu} \right] \nonumber\\
                                                                                                    & = \Prob \left[ \frac{|U|}{\snorm{U}_{\Psi_{\nu, L}}} \geq \Psi_{\nu, L}^{-1} \left( e^t - 1 \right)  \right] \nonumber\\
                                                                                                    & = \Prob \left[ \Psi_{\nu, L} \left(\frac{|U|}{\snorm{U}_{\Psi_{\nu, L}}} \right) + 1 \geq e^t \right] \nonumber\\
                                                                                                    & \leq \frac{\E \left[ \Psi_{\nu, L} \left( \frac{|U|}{\snorm{U}_{\Psi_{\nu, L}}}  \right) + 1 \right]}{e^{t}}  \leq 2e^{-t}.
                                                                                                    \label{eq:GBO_prob}
  \end{align}

  The following theorem (Theorem 3.1 of \citet{kuchibhotla2018moving}, simplified to i.i.d. case) gives an upper bound on the GBO norm of $\sum_i U_i / n$ for i.i.d. $U_i$:
  \begin{lemma}
    Consider i.i.d. $\subw(\nu)$ random variables $U_1, ..., U_n$, the following bound holds:
    \begin{align}
      \norm{\frac{1}{n} \sum_{i = 1}^n U_i - \E[U_1]}_{\Psi_{\nu, L_n(\nu)}} \leq \frac{2ec(\nu)}{\sqrt{n}}  \norm{U_1}_{\psi_{\nu}},
      \label{eq:GBO_upperbound}
    \end{align}
    where the constant $c(\nu) > 0$ only depends on $\nu$, and
    \begin{align}
      L_n(\nu) = \frac{4^{1 / \nu}}{\sqrt{2}} \times \begin{cases}
        \frac{1}{\sqrt{n}}  \qquad & 0 < \nu < 1,\\
        \frac{4e}{c(\nu) \sqrt{n}}  & \nu = 1.
      \end{cases}
      \label{eq:Ln}
    \end{align}
  \end{lemma}
  Combining \eqref{eq:GBO_prob} and \eqref{eq:GBO_upperbound}, we have
  \begin{align}
    \Prob \left( \left| \frac{1}{n} \sum_{i = 1}^n U_i - \E [U_1] \right| \geq 2 e c(\nu) \norm{U_1}_{\psi_\nu}  \left( \sqrt{\frac{t}{n}} + L_n(\nu) \frac{t^{1 / \nu}}{\sqrt{n}} \right) \right) \leq 2 e^{-t}.
    \label{eq:concentration2}
  \end{align} 

  The concentration inequality above shows that for small values of $t$, the tail bound for sub-Weibull averages behaves like a Gaussian (i.e., having $\sqrt{t/ n}$ tail), and for larger values of $t$, it has a much heavier tail.

  Finally, the single mixture bound in \eqref{eq:concentration} that holds for all values of $t$ follows from applying the inequality that $a + b \leq \sqrt{4ab}$ for any $a, b > 0$ in \eqref{eq:concentration2}.
\end{proof}

\subsection{Some concentration inequalities for interactions} \label{app:example_subweibull}
First let $\snorm{X}_{\psi_2}$ denote the sub-Gaussian norm of the sub-Gaussian random vector $X = (X_1, ..., X_p)$. For any $j, k, m, l \in [p]$, by Young's inequality, we have
\begin{align}
  &\E \left[ \exp \left( \frac{|X_j X_k X_m|^{2/3}}{\norm{X}_{\psi_{2}}^{2/3} \snorm{X}_{\psi_2}^{2/3} \snorm{X}_{\psi_2}^{2/3}}  \right) \right] \leq \E \left[ \exp \left( \frac{X_j^2}{3 \snorm{X}_{\psi_2}^2} + \frac{X_k^2}{3 \snorm{X}_{\psi_2}^2} + \frac{X_m^2}{3 \snorm{X}_{\psi_2}^2}    \right) \right] \nonumber\\
  &\leq \frac{1}{3} \E \left[ \exp \left( \frac{X_j^2}{\snorm{X}_{\psi_2}^2}  \right) \right] + \frac{1}{3} \E \left[ \exp \left( \frac{X_k^2}{\snorm{X}_{\psi_2}^2}  \right) \right] +  \frac{1}{3} \E \left[ \exp \left( \frac{X_m^2}{\snorm{X}_{\psi_2}^2}  \right) \right] \leq 3 \times \frac{2}{3} = 2.
  \nonumber
\end{align}
As a result, $X_jX_k X_m$ is a $\subw(2/3)$ random variable with $\snorm{X_jX_k X_m}_{\Psi_{2/3}} \leq \snorm{X}_{\psi_2}^3$.
Using similar arguments, we can show that $X_jX_kX_mX_l \sim \subw(1/2)$ random variable with $\snorm{X_jX_kX_mX_l}_{\Psi_{1/2}} \leq \snorm{X}_{\psi_2}^4$. Moreover, one can show that $W^T \gamma^\ast \sim \subw(1)$  and $(W^T \gamma^\ast)^2 \sim \subw(1/2)$.

Applying Theorem \ref{thm:GBO}, we have the following useful concentration inequalities for the products of main effects and interactions: 
\begin{corollary}
  Under Assumption \ref{a:distn}, for any $t > 0$,
  \begin{align}
    &\Prob \left( \left| \frac{1}{n} \sum_{i = 1}^n \X_{ij} \X_{ik} \X_{im} - \E [X_j X_k X_m] \right| \geq C(2/3) \norm{X}_{\psi_2}^3 \frac{t}{n^{3/4}} \right)   \leq 2 e^{-t},
    \nonumber\\
    &\Prob \left( \left| \frac{1}{n} \sum_{i = 1}^n \X_{ij} \X_{ik} \X_{im} \X_{il} - \E [X_j X_k X_m X_l] \right| \geq C(1/2) \norm{X}_{\psi_2}^4 \frac{t^{5/4}}{n^{3/4}} \right)   \leq 2 e^{-t}.
    \nonumber
  \end{align} 
\end{corollary}
Note that the inequalities above can be easily adapted to derive concentration inequalities for the products of the pure interaction $W$.

\section{Proof of Theorem \ref{thm:screening}} \label{app:proof_screening}
We follow the analysis in \citet{barut2015conditional} and \citet{2016arXiv160508933F}. First we let the vector $\onevec_n$ stands for a vector of $n$ ones, and $\C_n = \iden_n - \onevec_n \onevec_n^T / n$ is the centering matrix.  
We consider 
\begin{align}
  \omega_\ell = \frac{\frac{1}{n} \Z_\ell^T \C_n \r}{\sqrt{\frac{1}{n} \Z_\ell^T \C_n \Z_\ell}} = n^{-1/2} \norm{\C_n \Z_\ell}_2^{-1} \Z_\ell^T \C_n \left( \W \gamma^\ast + \X \theta^\ast - \X \hat{\theta} + \bvarepsilon \right),
  \label{eq:omegaj}
\end{align}
and the corresponding population quantity 
\begin{align}
  \omega^\ast_\ell = \frac{\Cov \left( Z_\ell, W^T \gamma^\ast \right)}{\sqrt{\Psi_{\ell\ell}}} = \frac{\Omega_\ell^T \gamma^\ast}{\sqrt{\Psi_{\ell\ell}}}.
  \label{eq:omegaj_star}
\end{align}

We first show that $\omega^\ast_\ell$ is useful in representing interaction variables $\ell \in \I(\alpha)$, and furthermore that $\omega_\ell$ converges to  $\omega^\ast_\ell$. As a result, we can use $\omega_\ell$, which is computable, as a noisy proxy for $\omega^\ast_\ell$ to determine whether $\ell$ is in $\I(\alpha)$. We formally present it as the following lemma
\begin{lemma} \label{lem:screening}
  Under Assumption \ref{a:distn} and \ref{a:dim} and with $\bar{\eta}^\ast$ as in \eqref{eq:eta_star_bar},
  \begin{align}
    \Prob\left( \max_{1 \leq \ell \leq q} |\omega_\ell - \omega^\ast_\ell| \leq \bar{\eta}^\ast \right) \geq 1 - 8 p^{-2(\kappa^{3/5} - 1)} - 2 p^{-1} - \Prob(\Event_R^C). 
    \nonumber
  \end{align}
\end{lemma}
\begin{proof}
  Denote $\tilde{\Z}_\ell = n^{1/2} \snorm{\C_n \Z_\ell}_2^{-1} \Z_\ell$, then we have that $\snorm{\C_n \tilde{\Z}_\ell}_2 = n^{1/2}$. From \eqref{eq:omegaj} and \eqref{eq:omegaj_star},
  \begin{align}
    &\max_{1 \leq \ell \leq q}|\omega_\ell - \omega^\ast_\ell| 
    = \max_{1 \leq \ell \leq q}\left|\frac{1}{n} \tilde{\Z}_\ell^T \C_n \left( \W \gamma^\ast + \X\theta^\ast - \X\hat{\theta} + \bvarepsilon \right) - \frac{\Omega_\ell^T \gamma^\ast}{\sqrt{\Psi_{\ell\ell}}} \right| \nonumber\\
    =& \max_{1 \leq \ell \leq q}\left|\left( \frac{1}{n} \tilde{\Z}_\ell^T \C_n \W - \Psi_{\ell\ell}^{-1/2}\Omega_\ell^T \right) \gamma^\ast + \frac{1}{n} \tilde{\Z}_\ell^T \C_n \X(\theta^\ast - \hat{\theta}) + \frac{1}{n} \tilde{\Z}_\ell^T \C_n \bvarepsilon \right| \nonumber\\
    \leq & \max_{1 \leq \ell \leq q} \left|\frac{1}{n} \tilde{\Z}_\ell^T \C_n \W \gamma^\ast - \Psi_{\ell\ell}^{-1/2}\Omega_{\ell}^T \gamma^\ast \right| + \max_{1 \leq \ell \leq q} \frac{1}{n} \norm{\C_n \tilde{\Z}_\ell}_2 \norm{\X \theta^\ast - \X \hat{\theta}}_2 +   \max_{1 \leq \ell \leq q} \frac{1}{n} \left|\tilde{\Z}_\ell^T \C_n \bvarepsilon \right|  \nonumber \\
    = & \max_{1 \leq \ell \leq q} \left|\frac{1}{n} \tilde{\Z}_\ell^T \C_n \W \gamma^\ast - \Psi_{\ell\ell}^{-1/2}\Omega_{\ell}^T \gamma^\ast \right| + \frac{1}{\sqrt{n}} \norm{\X \theta^\ast - \X \hat{\theta}}_2 + \max_{1\leq \ell \leq q} \frac{1}{n} \left| \tilde{\Z}_\ell^T \C_n\bvarepsilon \right|.
    \label{eq:term_all}
  \end{align}

  Furthermore we denote $\check{\Z}_\ell = \Psi_{\ell\ell}^{-1 / 2} \Z_\ell$.
  For any $\ell$ and $m$,
  \begin{align}
    &\left| \frac{1}{n} \tilde{\Z}_\ell^T \C_n \W \gamma^\ast - \Psi_{\ell\ell}^{-1/2}\Omega_{\ell}^T \gamma^\ast \right| = \left| \frac{n^{-1} \Z_\ell^T \C_n \W \gamma^\ast}{n^{-1/2} \snorm{\C_n \Z_\ell}_2} - \frac{\Omega_{\ell }^T \gamma^\ast}{\sqrt{\Psi_{\ell\ell}}}  \right| \nonumber\\
    =&\left|\frac{\left(n^{-1}\Z_\ell^T \C_n \W \gamma^\ast - \Omega_{\ell}^T \gamma^\ast \right) \sqrt{\Psi_{\ell\ell}} + \Omega_{\ell}^T \gamma^\ast \left( \sqrt{\Psi_{\ell\ell}} - n^{-1/2} \snorm{\C_n \Z_\ell}_2\right)}{n^{-1/2} \sqrt{\Psi_{\ell\ell}}\snorm{\C_n \Z_\ell}_2}  \right| \nonumber\\
    \leq& \left|n^{-1} \check{\Z}_\ell^T \C_n \W \gamma^\ast - \frac{\Omega_{\ell }^T \gamma^\ast}{\sqrt{\Psi_{\ell\ell}}}\right|\sqrt{n}\snorm{\C_n \check{\Z}_\ell}_2^{-1} + \frac{\Omega_{\ell}^T \gamma^\ast}{\sqrt{\Psi_{\ell\ell}}} \left| \sqrt{n} \norm{\C_n \check{\Z}_\ell}_2^{-1} - 1 \right|. \label{eq:t1}
  \end{align}

  Denote $\check{Z} = \diag(\Psi)^{-1/2} Z$, we can check that $Y_\ell = (\check{Z}_\ell - \E[\check{Z}_\ell])(W^T\gamma^\ast - E[W^T \gamma^\ast])$ is a $\subw(1/2)$ with $E[Y_\ell] = \frac{\Cov(Z_\ell, W^T \gamma^\ast)}{\sqrt{\Psi_{\ell \ell}}} = \frac{\Omega_\ell^T \gamma^\ast}{\sqrt{\Psi_{\ell \ell}}} $.
  Furthermore, we have $\snorm{Y_\ell}_{\psi_{1/2}} \leq c_1 \snorm{\diag(\Psi)^{-1/2}Z}_{\psi_1} \snorm{W^T \gamma^\ast}_{\psi_1}$ for some constant $c_1 > 0$.
  Therefore, by Theorem \ref{thm:GBO}, with constant $C_1 > 0$, we have
  \begin{align}
    \Prob \left( \left| n^{-1} \check{\Z}_\ell^T \C_n \W \gamma^\ast - \frac{\Omega_{\ell }^T \gamma^\ast}{\sqrt{\Psi_{\ell\ell}}} \right| \geq C_1 \snorm{\diag(\Psi)^{-1/2}Z}_{\psi_1}\snorm{W^T \gamma^\ast}_{\psi_1} \frac{t^{5/4}}{n^{3/4}} \right)   \leq 2 e^{-t}.
    \label{eq:t11}
  \end{align}  

  Similarly, $(\check{Z}_\ell - \E[\check{Z}_\ell])^2$ is a $\subw(1/2)$ random variable with $\E (\check{Z}_\ell - \E[\check{Z}_\ell])^2 = \frac{\Psi_{\ell \ell}}{\Psi_{\ell \ell}} = 1$, and $\snorm{(\check{Z}_\ell - \E[\check{Z}_\ell])^2}_{\psi_{1/2}} \leq c_2\snorm{\check{Z}_\ell}_{\psi_1}^2 \leq c_2$ for some constant $c_2 > 0$.
  By Theorem \ref{thm:GBO}, with some constant $C_2 > 0$, we have
  \begin{align}
    \Prob \left( \left| n^{-1} \norm{\C_n\check{\Z}_\ell}_2^2 - 1 \right| \geq \frac{C_2}{2} \frac{t^{5/4}}{n^{3/4}} \right)   \leq 2 e^{-t}.
    \nonumber
  \end{align}  
  For any $\epsilon > 0$, 
  \begin{align}
    &\Prob \left( \left| n^{1/2} \snorm{\C_n \check{\Z}_\ell}_2^{-1} - 1 \right| \geq \epsilon \right) \nonumber\\
    =& \Prob \left( \left| n^{1/2} \snorm{\C_n \check{\Z}_\ell}_2^{-1} - 1 \right| \geq \epsilon, n^{-1} \snorm{\C_n \check{\Z}_\ell}_2^2 \leq 1 + \epsilon\right)
    + \Prob \left( \left| n^{1/2} \snorm{\C_n \check{\Z}_\ell}_2^{-1} - 1 \right| \geq \epsilon, n^{-1} \snorm{\C_n \check{\Z}_\ell}_2^2 > 1 + \epsilon\right) \nonumber\\
    \leq & \Prob \left( n^{-1} \snorm{\C_n \check{\Z}_\ell}_2^2 \leq 1 + \epsilon\right)
    + \Prob \left( \frac{\left| n^{-1} \snorm{\C_n \check{\Z}_\ell}_2^{2} - 1 \right|}{n^{-1/2} \snorm{\C_n \check{\Z}_\ell}_2 \left( n^{-1/2} \snorm{\C_n \check{\Z}_\ell}_2 + 1 \right)} \geq \epsilon, n^{-1} \snorm{\C_n \check{\Z}_\ell}_2^2 > 1 + \epsilon\right) \nonumber\\
    \leq & \Prob \left( 1 - \left| n^{-1} \snorm{\C_n \check{\Z}_\ell}_2^2 - 1 \right| \leq 1 + \epsilon\right) + \Prob \left( \left| n^{-1} \snorm{\C_n \check{\Z}_\ell}_2^{2} - 1 \right| \geq \epsilon \right) \nonumber\\
    = & 2 \Prob \left( \left| n^{-1} \snorm{\C_n \check{\Z}_\ell}_2^2 - 1 \right| \geq \epsilon\right). \nonumber
  \end{align}
  Take $\epsilon = 2^{-1} C_2 t^{5/4}n^{-3/4}$, we have
  \begin{align}
    \Prob \left( \left| n^{1/2} \snorm{\C_n \check{\Z}_\ell}_2^{-1} - 1 \right| \geq \frac{C_2}{2} \frac{t^{5/4}}{n^{3/4}}  \right) \leq 4 e^{-t}.
    \label{eq:t12}
  \end{align}
  First set $t = 2 n^{3/5}$ in \eqref{eq:t12}, we have that
  \begin{align}
    \Prob \left( n^{1/2} \snorm{\C_n \check{\Z}_\ell}_2^{-1} \geq 1 + C_2  \right) \leq 4 e^{-2n^{3/5}}.
    \label{eq:t13}
  \end{align}

  Combining \eqref{eq:t11}, \eqref{eq:t12}, and \eqref{eq:t13}, from \eqref{eq:t1} and union bounds we have that
  \begin{align}
    &\Prob \left\{ \frac{1}{n}\max_{1 \leq \ell \leq q} \left|  \tilde{\Z}_\ell^T \C_n \W \gamma^\ast - \Psi_{\ell\ell}^{-1/2}\Omega_{\ell}^T \gamma^\ast \right| 
    \geq \left[C_1 (1 + C_2)\snorm{\diag(\Psi)^{-1/2}Z}_{\psi_1} \snorm{W^T \gamma^\ast}_{\psi_1} + C_2 \max_{\ell} |\omega_\ell^\ast| \right] \frac{t^{5/4}}{n^{3/4}}   \right\} \nonumber\\
    \leq& 4 \exp \left(  2 \log p - t \right) + 4 \exp \left( 2 \log p - 2 n^{3/5} \right).
    \nonumber
  \end{align}
  Take $t = 2(\log p)^{3/5} n^{1/5}$, we have
  \begin{align}
    &\Prob \left\{ \frac{1}{n}\max_{1 \leq \ell \leq q} \left|  \tilde{\Z}_\ell^T \C_n \W \gamma^\ast - \Psi_{\ell\ell}^{-1/2}\Omega_{\ell}^T \gamma^\ast \right| 
    \geq K_1 \left[\snorm{\diag(\Psi)^{-1/2}Z}_{\psi_1} \snorm{W^T \gamma^\ast}_{\psi_1} + \max_{\ell} |\omega_\ell^\ast| \right] \frac{(\log p)^{3/4}}{n^{1/2}}   \right\} \nonumber\\
    \leq& 4 \exp \left(  2 \log p - 2 (\log p)^{3/5} n^{1/5} \right) + 4 \exp \left( 2 \log p - 2 n^{3/5} \right).
    \label{eq:term_first}
  \end{align}
  where $K_1 = 2\max \left\{ C_1 (1 + C_2), C_2 \right\}$.

  For each $1 \leq \ell \leq q$, $\bvarepsilon$ and $\C_n \tilde{\Z}_\ell$ are independent, and $\E(\bvarepsilon^T \C_n \tilde{\Z}_\ell) = \E(\bvarepsilon)^T \E(\C_n\tilde{\Z}_\ell) = 0$.
  Recall that $\snorm{\C_n \tilde{\Z}_\ell}_2 = n^{1/2}$, conditional on which $\bvarepsilon^T \C_n \tilde{\Z}_\ell$ follows a sub-Gaussian distribution ($\subw(2)$) with mean zero and variance $\sigma^2 \snorm{\C_n\tilde{\Z}_\ell}_2^2 = n \sigma^2$. 
  From a Hoeffding-type inequality \citep[see, e.g.,][]{vershynin2010introduction} we have 
  \begin{align}
      &\Prob\left( \frac{1}{n} \max_{1 \leq \ell \leq q}\left| \bvarepsilon^T \C_n \tilde{\Z}_\ell \right| > K_2 \sigma \sqrt{\frac{\log p}{n}} \right) \nonumber\\
    \leq&
    \Prob\left( \frac{1}{n} \max_{1 \leq \ell \leq q}\left| \bvarepsilon^T \C_n \tilde{\Z}_\ell \right| > K_2 \sigma \sqrt{\frac{\log p}{n}} \Big | \snorm{\C_n \tilde{\Z}_\ell}_2 =  \sqrt{n}\right) +  \Prob \left(\snorm{\C_n \tilde{\Z}_\ell}_2 \neq  \sqrt{n} \right)  \nonumber\\
    = & \Prob\left( \frac{1}{n} \max_{1 \leq \ell \leq q}\left| \bvarepsilon^T \C_n \tilde{\Z}_\ell \right| > K_2 \sigma \sqrt{\frac{\log p}{n}} \Big | \snorm{\C_n \tilde{\Z}_\ell}_2 =  \sqrt{n}\right) \nonumber\\
    \leq& 2 \exp \left( 2 \log p - 3 \log p  \right) = 2 \exp \left( - \log p \right)
    \label{eq:term_second}
  \end{align}
  for some constant $K_2 > 0$.

  Now combining \eqref{eq:term_first} and \eqref{eq:term_second}, we have from \eqref{eq:term_all} that
  \begin{align}
      &\max_{1 \leq \ell \leq q} |\omega_\ell - \omega^\ast_\ell| \leq  K_1 \left[\snorm{\diag(\Psi)^{-1/2}Z}_{\psi_1} \snorm{W^T \gamma^\ast}_{\psi_1} + \max_{\ell} |\omega^\ast_\ell| \right] \frac{(\log p)^{3/4}}{n^{1/2}} + R  + K_2 \sigma \sqrt{\frac{\log p}{n}} \nonumber\\
    \leq  &  \frac{K}{2}\left[ \left(\snorm{\diag(\Psi)^{-1/2}Z}_{\psi_1}\snorm{W^T \gamma^\ast}_{\psi_1} + \max_{\ell} |\omega^\ast_\ell| \right) \frac{(\log p)^{3/4}}{n^{1/2}} + R  + \sigma \sqrt{\frac{\log p}{n}} \right] := \bar{\eta}^\ast
    \nonumber
  \end{align}
  holds with probability greater than $1 - 4 \exp \left(  2 \log p - 2 (\log p)^{3/5} n^{1/5} \right) - 4 \exp \left( 2 \log p - 2 n^{3/5} \right) - 2 p^{-1} - \Prob(\Event_R^C)$, where $K = \max \left\{ K_1, 1, K_2 \right\}$.
  Finally the results follows from Assumption \ref{a:dim} that $\kappa \log p \leq n^{1/2} < n$.
\end{proof}

Now for any $\eta \in [\bar{\eta}^\ast, \eta(\alpha)]$, consider the following event 
\begin{align}
  \mathcal{E} = \left\{ \max_{\ell \in \I(\alpha)} |\omega_\ell - \omega^\ast_\ell| \leq \eta \right\}.
  \nonumber
\end{align}
Recall from \eqref{eq:eta_alpha} that $\min_{\ell \in \I(\alpha)} \omega_\ell^\ast =  2  \eta(\alpha) $. For any $\ell \in \I(\alpha)$, by assumption \eqref{eq:eta_star} and on event $\mathcal{E}$, we have
\begin{align}
  |\omega_\ell| \geq |\omega^\ast_\ell| - |\omega^\ast_\ell - \omega_\ell| > \min_{\ell \in \I(\alpha)} |\omega_\ell^\ast| - \eta = 2 \eta(\alpha) - \eta \geq \eta,
  \nonumber
\end{align}
which implies that $\ell \in \hat{\I}_{\eta}$. Thus
\begin{align}
  \Prob\left( \I(\alpha) \subseteq \hat{\I}_{\eta} \right) \geq \Prob\left( \mathcal{E} \right) \geq 1 - \Prob\left( \max_{\ell \in \I(\alpha)} |\omega_\ell - \omega^\ast_\ell| > \bar{\eta}^\ast \right).
  \nonumber
\end{align}

To show the second half of Theorem \ref{thm:screening}, we first give an upper bound on $\sum_{\ell = 1}^q {\omega^\ast_\ell}^2$. First note that
\begin{align}
  \sum_{\ell = 1}^q {\omega^\ast_\ell}^2 = \sum_{\ell = 1}^q \Psi_{\ell\ell}^{-1}(\Omega_\ell^T \gamma^\ast)^2 &=  \norm{\diag(\Psi)^{-1/2}\Omega \gamma^\ast}_2^2 =  \norm{\diag(\Psi)^{-1/2}\Omega^{1/2} \Omega^{1/2}\gamma^\ast}_2^2 \nonumber\\
                                                                                                               &\leq \lambda_{\max}\left( \diag(\Psi)^{-1/2}\Omega\diag(\Psi)^{-1/2} \right) {\gamma^\ast}^T \Omega \gamma^\ast,
                                                                                                               \nonumber
\end{align}
and that 
\begin{align}
  \Var \left( W^T \gamma^\ast \right) = {\gamma^\ast}^T \Omega \gamma^\ast,
  \nonumber
\end{align}
which together imply that $\sum_{\ell = 1}^q {\omega^\ast_\ell}^2 \leq \lambda_{\max}\left( \diag(\Psi)^{-1/2}\Omega\diag(\Psi)^{-1/2} \right) \Var(W^T \gamma^\ast)$.
Consider the set $\tilde{\I}_\eta = \{ \ell: |\omega^\ast_\ell| > 2^{-1} \eta \}$.
Conditional on $\mathcal{E}$, for any $\ell \in \hat{\I}_{\eta}$, we have that
\begin{align}
  |\omega^\ast_\ell| \geq |\omega_\ell| - |\omega_\ell - \omega^\ast_\ell| > \eta - 2^{-1} \eta = 2^{-1} \eta,
  \nonumber
\end{align}
which implies that $\ell \in \tilde{\I}_\eta$ and thus $\hat{\I}_\eta \subseteq \tilde{\I}_\eta$. Finally
\begin{align}
  |\hat{\I}_{\eta}| \leq |\tilde{\I}_\eta| \leq \frac{4\lambda_{\max}\left( \diag(\Psi)^{-1/2}\Omega\diag(\Psi)^{-1/2} \right) \Var(W^T \gamma^\ast)}{{\eta}^2}.
  \nonumber
\end{align}

\section{Proof of Theorem \ref{thm:prediction}} \label{app:proof_prediction}
\begin{proof}
  Recall that
  \begin{align}
    \r = \y - \X \hat{\theta} = \X (\theta^\ast - \hat{\theta}) + \W \gamma^\ast + \bvarepsilon.
    \nonumber
  \end{align}
  The basic inequality of \eqref{eq:step3} then implies that
  \begin{align}
  &\frac{1}{2n} \norm{\r - \X \hat{\xi} - \Z_{\hat{\I}_\eta} \hat{\varphi}}_2^2 + \lambda\left( \snorm{\hat{\xi}}_1 + \snorm{\hat{\varphi}}_1 \right) \nonumber\\
    = &\frac{1}{2n} \norm{\X \theta^\ast - \X \hat{\theta}  + \W \gamma^\ast - \X \hat{\xi} - \Z_{\hat{\I}_\eta} \hat{\varphi} + \bvarepsilon}_2^2 + \lambda\left( \snorm{\hat{\xi}}_1 + \snorm{\hat{\varphi}}_1 \right)
    \nonumber\\
    \leq &\frac{1}{2n} \norm{\X \theta^\ast - \X \hat{\theta}  + \W \gamma^\ast - \X \bar{\xi} - \Z_{\hat{\I}_\eta} \bar{\varphi} + \bvarepsilon}_2^2 + \lambda\left( \snorm{\bar{\xi}}_1 + \snorm{\bar{\varphi}}_1 \right)
    \nonumber
  \end{align}
  for any pair of $\bar{\xi} \in \real^p$ and $\bar{\varphi} \in \real^{|\hat{\I}_\eta|}$.
  We have that
  \begin{align}
  &\frac{1}{2n} \norm{\X \theta^\ast - \X \hat{\theta} + \W \gamma^\ast - \X \hat{\xi} - \Z_{\hat{\I}_\eta} \hat{\varphi}}_2^2 + \lambda\left( \snorm{\hat{\xi}}_1 + \snorm{\hat{\varphi}}_1 \right) 
  \nonumber\\
    \leq & \frac{1}{n} \bvarepsilon^T \X \left( \hat{\xi} - \bar{\xi} \right) + \frac{1}{n} \bvarepsilon^T \Z_{\hat{\I}_\eta} \left( \hat{\varphi} - \bar{\varphi} \right) + \frac{1}{2n} \norm{\X \theta^\ast - \X \bar{\theta} + \W \gamma^\ast - \X \bar{\xi} -  \Z_{\hat{\I}_\eta} \bar{\varphi}}_2^2 + \lambda\left( \norm{\bar{\xi}}_1 + \norm{\bar{\varphi}}_1 \right) \nonumber\\
    \leq &\frac{1}{n} \max_j |\bvarepsilon^T \X_j| \left( \snorm{\hat{\xi}}_1 + \snorm{\bar{\xi}}_1 \right) + \frac{1}{n} \max_{\ell \in \hat{\I}_\eta} |\bvarepsilon^T  \Z_{\ell}| \left(\snorm{\hat{\varphi}}_1 + \snorm{\bar{\varphi}}_1 \right) \nonumber\\
         &+ \frac{1}{2n} \norm{\X \theta^\ast - \X \hat{\theta} + \W \gamma^\ast - \X \bar{\xi} - \Z_{\hat{\I}_\eta} \bar{\varphi}}_2^2 + \lambda \left( \snorm{\bar{\xi}}_1 + \snorm{\bar{\varphi}}_1 \right).
         \nonumber
  \end{align}
  Consider the events
  \begin{align}
    \T_1 = \left\{ \frac{1}{n} \max_{1 \leq j \leq p} |\bvarepsilon^T \X_j| \leq \lambda \right\} \quad \quad
    \T_2 = \left\{ \frac{1}{n} \max_{\ell \in \hat{\I}_\eta} |\bvarepsilon^T \Z_\ell| \leq \lambda \right\}.
    \label{eq:events_What_rinter}
  \end{align}
  On $\T_1 \cap \T_2$, we have
  \begin{align}
    \frac{1}{2n} \norm{\X \theta^\ast - \X \hat{\theta} + \W \gamma^\ast -  \X \hat{\xi} - \Z_{\hat{\I}_\eta}\hat{\varphi}}_2^2 
    \leq \frac{1}{2n} \norm{\X \theta^\ast - \X \hat{\theta} + \W \gamma^\ast - \X \bar{\xi} - \Z_{\hat{\I}_\eta} \bar{\varphi}}_2^2 + 2\lambda \left( \snorm{\bar{\xi}}_1 + \snorm{\bar{\varphi}}_1 \right).
    \nonumber
  \end{align}
  For any $\alpha \geq 0$ and suppose that $\I(\alpha) \subseteq \hat{\I}_\eta$, we define $\bar{\xi} = - \Sigma^{-1} \Phi_{\I(\alpha)} \gamma^\ast_{\I(\alpha)}$ and \begin{align}
    \bar{\varphi}_\ell = \begin{cases}
      &\gamma^\ast_\ell \quad \ell \in \I(\alpha) \\
      &0 \quad \ell \in \hat{\I}_\eta \setminus \I(\alpha).
    \end{cases}
    \nonumber
  \end{align}
  Recall from Section \ref{sec:principle} that $\Z = \W + \X \Sigma^{-1} \Phi$ and $\theta^\ast = \beta^\ast + \Sigma^{-1} \Phi \gamma^\ast$, we have
  \begin{align}
    \X \bar{\xi} = - \X \Sigma^{-1} \Phi_{\I(\alpha)} \gamma^\ast_{\I(\alpha)},
    \nonumber
  \end{align}
  and
  \begin{align}
    \Z_{\hat{\I}_\eta} \bar{\varphi} = \Z_{\I(\alpha)} \gamma^\ast_{\I(\alpha)} = \W_{\I(\alpha)} \gamma^\ast_{\I(\alpha)} + \X \Sigma^{-1} \Phi_{\I(\alpha)} \gamma^\ast_{\I(\alpha)} = \W_{\I(\alpha)} \gamma^\ast_{\I(\alpha)} - \X \bar{\xi}.
    \nonumber
  \end{align}
  Then
  \begin{align}
    \frac{1}{2n} \norm{\X \theta^\ast - \X \hat{\theta} + \W \gamma^\ast - \X \bar{\xi} - \Z_{\hat{\I}_\eta} \bar{\varphi}}_2^2 = \frac{1}{2n} \norm{\X \theta^\ast - \X \hat{\theta} + \W_{\I(\alpha)^C} \gamma^\ast_{\I(\alpha)^C}}_2^2.
    \nonumber
  \end{align}
  Therefore, on $\T_1 \cap \T_2$,
  \begin{align}
    &\frac{1}{2n} \norm{\X \theta^\ast - \X \hat{\theta} + \W \gamma^\ast - \X \hat{\xi} - \Z_{\hat{\I}_\eta} \hat{\varphi}}_2^2 \nonumber\\
    \leq& \frac{1}{2n} \norm{\X \theta^\ast - \X \hat{\theta} + \W \gamma^\ast - \X \bar{\xi} - \Z_{\hat{\I}_\eta} \bar{\varphi}}_2^2 + 2\lambda \left( \snorm{\bar{\xi}}_1 + \snorm{\bar{\varphi}}_1 \right) \nonumber\\
    =& \frac{1}{2n} \norm{\X \theta^\ast - \X \hat{\theta} + \W_{\I(\alpha)^C} \gamma^\ast_{\I(\alpha)^C}}_2^2 + 2\lambda \left( \snorm{\Sigma^{-1} \Phi_{\I(\alpha)} \gamma^\ast_{\I(\alpha)}}_1 + \snorm{\gamma^\ast_{\I(\alpha)}}_1 \right) \nonumber\\
    \leq& \frac{1}{n} \norm{\X \theta^\ast - \X \hat{\theta}}_2^2 + \frac{1}{n} \norm{\W_{\I(\alpha)^C} \gamma^\ast_{\I(\alpha)^C}}_2^2 + 2\lambda \left( \snorm{\Sigma^{-1} \Phi_{\I(\alpha)} \gamma^\ast_{\I(\alpha)}}_1 + \snorm{\gamma^\ast_{\I(\alpha)}}_1 \right) \nonumber.
  \end{align}

\end{proof}

\section{Proof of Theorem \ref{thm:lasso_step1}} \label{app:proof_step1}
We start from the basic inequality that
\begin{align}
  \frac{1}{2n} \snorm{\y - \X \check{\theta}}_2^2 + \lambda \snorm{\check{\theta}}_1 \leq \frac{1}{2n} \snorm{\y - \X \theta^\ast}_2^2 + \lambda \snorm{\theta^\ast}_1,
  \nonumber
\end{align}
which implies that
\begin{align}
  \frac{1}{2n} \snorm{\X \check{\theta} - \X \theta^\ast}_2^2 + \lambda \snorm{\check{\theta}}_1 \leq \frac{1}{n} \left( \check{\theta} - \theta^\ast \right)^T \X^T \left( \W \gamma^\ast + \varepsilon \right) + \lambda \snorm{\theta^\ast}_1.
  \nonumber
\end{align}
The ``empirical process'' part can be bounded by
\begin{align}
  \frac{1}{n} \left| \left( \check{\theta} - \theta^\ast \right)^T \X^T \left( \W \gamma^\ast + \bvarepsilon \right) \right|
  \leq \frac{1}{n} \max_{1 \leq j \leq p} \left| \X_j^T \left( \W \gamma^\ast + \bvarepsilon \right) \right| \snorm{\check{\theta} - \theta}_1 .
  \nonumber
\end{align}
Denote the event 
\begin{align}
  \T = \left\{\frac{1}{n} \max_{1 \leq j \leq p} \left| \X_j^T \left( \W \gamma^\ast + \bvarepsilon \right) \right| \leq \lambda_0 \text{ for some } \lambda_0 > 0 \right\}.
  \nonumber
\end{align}
Then on $\T$, for any $\lambda \geq \lambda_0$,
\begin{align}
  \frac{1}{2n} \snorm{\X \check{\theta} - \X \theta^\ast}_2^2 + \lambda \snorm{\check{\theta}}_1 \leq \lambda \snorm{\check{\theta} - \theta^\ast}_1 + \lambda \snorm{\theta^\ast}_1,
  \nonumber
\end{align}
which further implies the slow rate bound in prediction error, i.e., $\frac{1}{2n} \snorm{\X \check{\theta} - \X \theta^\ast}_2^2 \leq 2\lambda \snorm{\theta^\ast}_1$. 
We now characterize the scale of $\lambda_0$ and the probability that $\T$ holds:

For any $1 \leq j \leq p$,
\begin{align}
  \frac{1}{n} \left| \X_j^T \left( \W \gamma^\ast + \bvarepsilon \right) \right| \leq \frac{1}{n} \left| \X_j^T \W \gamma^\ast \right| + \frac{1}{n} \left| \X_j^T \bvarepsilon \right|.
  \nonumber
\end{align}
We start with $n^{-1} \X_j^T \bvarepsilon$. 
For any $j \in [p]$, $\varepsilon X_j \sim \subw(1)$, with $\E[\varepsilon X_j] = 0$ and $\snorm{\varepsilon X_j}_{\psi_1} \leq \sigma \snorm{X}_{\psi_2}$. So by Theorem \ref{thm:GBO} and a union bound, for any $t > 0$,
\begin{align}
  \Prob \left[ \frac{1}{n} \max_{1 \leq j \leq p}\left| \bvarepsilon^T \X_j \right| \geq C \left( 1  \right) \sigma \norm{X}_{\psi_2} \frac{t^{3/4}}{n^{3/4}} \right] \leq 2pe^{-t}.
  \nonumber
\end{align}
Take $t = (\log p)^{2/3} n^{1/3}$, we have
\begin{align}
  \Prob \left[ \frac{1}{n} \max_{1 \leq j \leq p}\left| \bvarepsilon^T \X_j \right| \geq C \left( 1  \right) \sigma \norm{X}_{\psi_2} \sqrt{\frac{\log p}{n}} \right] \leq 2\exp\left\{\log p - (\log p)^{2/3} n^{1/3} \right\}.
  \nonumber
\end{align}

Similarly, as shown in Appendix \ref{app:example_subweibull}, for each $j \in [p]$, $X_j W^T \gamma^\ast \sim \subw(2/3)$, with $\E[X_j W^T\gamma^\ast] = \Cov(X_j, W^T\gamma^\ast) + \E[X_j] \E[W^T\gamma^\ast] = 0$, and $\snorm{X_j W^T\gamma^\ast}_{\psi_{2/3}} \leq \snorm{X}_{\psi_2} \snorm{W^T \gamma^\ast}_{\psi_1}$.
So by Theorem \ref{thm:GBO} and a union bound, for any $t > 0$, 
\begin{align}
  \Prob \left( \left|\frac{1}{n} \X_j^T \W\gamma^\ast \right|  \geq C(2/3) \snorm{X}_{\psi_{2}} \snorm{W^T\gamma^\ast}_{\psi_1} \frac{t}{n^{3/4}} \right) \leq 2e^{-t}.
  \nonumber
\end{align}
Take $t = (\log p)^{1/2}n^{1/4}$, from a union bound we have
\begin{align}
  \Prob \left( \frac{1}{n} \max_{1 \leq j \leq p}\left|\X_j^T \W\gamma^\ast \right|  \geq C(2/3) \snorm{X}_{\psi_{2}} \snorm{W^T\gamma^\ast}_{\psi_1} \sqrt{\frac{\log p}{n}} \right) \leq 2\exp \left\{\log p - (\log p)^{1/2} n^{1/4}\right\}.
  \nonumber
\end{align}
Summarizing the result, we take
\begin{align}
  \lambda_0 = C \left( \sigma +\snorm{W^T\gamma^\ast}_{\psi_1} \right) \snorm{X}_{\psi_2} \sqrt{\frac{\log p}{n}},
  \nonumber
\end{align}
where $C = \max \left\{ C(1), C(2/3) \right\}$,
and a union bound implies that $\T$ holds with probability greater than $1 -2\exp\left\{\log p - (\log p)^{2/3} n^{1/3} \right\} -  2\exp \left\{\log p - (\log p)^{1/2} n^{1/4}\right\}$.

\section{Proof of Theorem \ref{thm:main}} \label{app:proof_thm_main}
First we note that $\I(\bar{\alpha}) \subseteq \I(\alpha)$ for any $\bar{\alpha} \geq \alpha$.
So if $\I(\alpha) \subseteq \hat{\I}_\eta$ holds for some $\bar{\alpha}$, then from Section \ref{app:proof_prediction},
\begin{align}
&\frac{1}{2n} \norm{\X \theta^\ast - \X \hat{\theta} + \W \gamma^\ast - \X \hat{\xi} - \Z_{\hat{\I}_\eta} \hat{\varphi}}_2^2 \nonumber\\
  \leq & R^2 + \inf_{\bar{\alpha} \geq \alpha} \left\{ \frac{1}{n} \norm{\W_{\I(\bar{\alpha})^C} \gamma^\ast_{\I(\bar{\alpha})^C}}_2^2 + 2\lambda \left( \snorm{\Sigma^{-1} \Phi_{\I(\bar{\alpha})} \gamma^\ast_{\I(\bar{\alpha})}}_1 + \snorm{\gamma^\ast_{\I(\bar{\alpha})}}_1 \right) \right\}.
  \nonumber
\end{align}

Define $\Event_2$ to be the event that \eqref{eq:efficiency} holds, and $\Event_1$ to be the event that \eqref{eq:bound_pred} holds.
We first find the value of $\lambda$ and the corresponding probability such that $\Event_2$ holds. 
Note that for each $\ell \in \hat{\I}_\eta$, it is easy to verify that $\varepsilon Z_\ell \sim \subw(2/3)$, with $\E[\epsilon Z_\ell] = 0$, and $\snorm{\varepsilon Z_\ell}_{\psi_{2/3}} \leq \sigma \snorm{Z}_{\psi_1}$.
So by Theorem \ref{thm:GBO} and a union bound, for any $t > 0$, 
\begin{align}
  \Prob \left[ \frac{1}{n} \max_{\ell \in \hat{\I}_\eta}\left| \bvarepsilon^T \Z_\ell \right| \geq C \left( 2/3  \right) \sigma \norm{Z}_{\psi_1} \frac{t}{n^{3/4}} \right] \leq 2 |\hat{\I}_\eta| e^{-t}.
  \nonumber
\end{align}
Take $t = 2(\log p)^{1/2} n^{1/4}$, we have
\begin{align}
  \Prob \left[ \frac{1}{n} \max_{\ell \in \hat{\I}_\eta}\left| \bvarepsilon^T \Z_\ell \right| \geq 2 C \left( 2/3  \right) \sigma \norm{Z}_{\psi_1} \sqrt{\frac{\log p}{n}} \right] &\leq 2 \exp \left\{\log|\hat{\I}_\eta| - 2 (\log p)^{1/2} n^{1/4} \right\} \nonumber\\
                                                                                                                                                                                    &\leq 2 \exp \left\{ - 2 (\sqrt{\kappa} - 1) \log p \right\},
                                                                                                                                                                                    \nonumber
\end{align}
where the last inequality holds because $\kappa \log p \leq \sqrt{n}$ from Assumption \ref{a:dim}.

Similarly, for any $j \in [p]$, $\varepsilon X_j \sim \subw(1)$, with $\E[\varepsilon X_j] = 0$ and $\snorm{\varepsilon X_j}_{\psi_1} \leq \sigma \snorm{X}_{\psi_2}$. So by Theorem \ref{thm:GBO} and a union bound, for any $t > 0$,
\begin{align}
  \Prob \left[ \frac{1}{n} \max_{1 \leq j \leq p}\left| \bvarepsilon^T \X_j \right| \geq C \left( 1  \right) \sigma \norm{X}_{\psi_2} \frac{t^{3/4}}{n^{3/4}} \right] \leq 2pe^{-t}.
  \nonumber
\end{align}
Take $t = 2(\log p)^{2/3} n^{1/3}$, we have
\begin{align}
  \Prob \left[ \frac{1}{n} \max_{1 \leq j \leq p}\left| \bvarepsilon^T \X_j \right| \geq 2C \left( 1  \right) \sigma \norm{X}_{\psi_2} \sqrt{\frac{\log p}{n}} \right] 
    &\leq 2\exp\left\{\log p - (\log p)^{2/3} n^{1/3} \right\} \nonumber\\
    &\leq 2\exp \left\{ - 2(\kappa^{1/3} - 1) \log p \right\}.
    \nonumber
\end{align}

Finally note that for any $\bar{\alpha} \geq \alpha$, $(W_{\I(\bar{\alpha})^C}^T \gamma^\ast_{\I(\bar{\alpha})^C})^2$ is a $\subw(1/2)$ random variable. By Triangle inequality and Lemma \ref{lem:exp},
\begin{align}
  \frac{1}{n}\norm{\W_{\I(\bar{\alpha})^C} \gamma^\ast_{\I(\bar{\alpha})^C}}_2^2 &\leq \left| \frac{1}{n}\norm{\W_{\I(\bar{\alpha})^C} \gamma^\ast_{\I(\bar{\alpha})^C}}_2^2 - \E \left[\left(W_{\I(\bar{\alpha})^C}^T \gamma^\ast_{\I(\bar{\alpha})^C} \right)^2  \right] \right| + \E \left[\left(W_{\I(\bar{\alpha})^C}^T \gamma^\ast_{\I(\bar{\alpha})^C}\right)^2 \right] \nonumber\\
                                                                                 &\leq \left| \frac{1}{n}\norm{\W_{\I(\bar{\alpha})^C} \gamma^\ast_{\I(\bar{\alpha})^C}}_2^2 - \E \left[\left(W_{\I(\bar{\alpha})^C}^T \gamma^\ast_{\I(\bar{\alpha})^C} \right)^2  \right] \right| + 4 \norm{ \left(W_{\I(\bar{\alpha})^C}^T \gamma^\ast_{\I(\bar{\alpha})^C}\right)^2}_{\psi_{1/2}} \nonumber\\
                                                                                 &\leq  \left| \frac{1}{n}\norm{\W_{\I(\bar{\alpha})^C} \gamma^\ast_{\I(\bar{\alpha})^C}}_2^2 - \E \left[\left(W_{\I(\bar{\alpha})^C}^T \gamma^\ast_{\I(\bar{\alpha})^C}\right)^2  \right] \right| + 4\bar{\alpha} \nonumber.
\end{align}
By Theorem \ref{thm:GBO}, we have that
\begin{align}
  \Prob \left[\left| \frac{1}{n}\norm{\W_{\I(\bar{\alpha})^C} \gamma^\ast_{\I(\bar{\alpha})^C}}_2^2 - \E \left[\left(W_{\I(\bar{\alpha})^C}^T \gamma^\ast_{\I(\bar{\alpha})^C}\right)^2  \right] \right| \geq C \left(1/2\right) \norm{\left(W_{\I(\bar{\alpha})^C}^T \gamma^\ast_{\I(\bar{\alpha})^C}\right)^2}_{\psi_{1/2}} \frac{t^{5/4}}{n^{3/4}} \right] \leq 2e^{-t}.
  \nonumber
\end{align}
Take $t = n^{3/5}$ we have
\begin{align}
  \Prob \left[\left| \frac{1}{n}\norm{\W_{\I(\bar{\alpha})^C} \gamma^\ast_{\I(\bar{\alpha})^C}}_2^2 - \E \left[\left(W_{\I(\bar{\alpha})^C}^T \gamma^\ast_{\I(\bar{\alpha})^C}\right)^2  \right] \right| \geq C \left(1/2\right) \bar{\alpha} \right] \leq 2e^{-n^{3/5}}.
  \nonumber
\end{align}

In summary, by a union bound and $\kappa^{1/2} > \kappa^{1/3}$, we have that
\begin{align}
  \Prob (\Event_1) \geq 1 - 4 p^{ -2 \left( \kappa^{1/3} - 1 \right)} - 2 \exp \left( -n^{3/5} \right),
  \nonumber
\end{align}
with $C_2 = 4(C(1/2) + 1)$.
By Theorem \ref{thm:screening}, Lemma \ref{lem:screening}, we have that
\begin{align}
  \Prob \left( \Event_2 \right) \geq 1 - 8 p^{-2(\kappa^{3/5} - 1)} - 2 p^{-1} - \Prob(\Event_R^C) \geq 1 - 8 p^{-2(\kappa^{1/3} - 1)} - 2 p^{-1} - \Prob(\Event_R^C). \nonumber
\end{align}
Finally, from Theorem \ref{thm:lasso_step1}, we plug in $R =\left(\sigma + \snorm{W^T\gamma^\ast}_{\psi_{1}} \right)^{1/2} \snorm{X}_{\psi_2}^{1/2} n^{-1/4}(\log p)^{1/4} \snorm{\theta^\ast}_2^{1/2}$, with $\Prob(\Event_R^C) \leq 4p^{-2(\kappa^{1/3} - 1)}$, and rearrange terms.
The probability result then follows a union bound on $\Prob (\Event_1 \cap \Event_2) = 1 - \Prob \left( \Event_1^C \cup \Event_2^C \right)$.

\section{Screening property of $\I^\top_k$} \label{app:screening_m}
In Section \ref{sec:method}, we introduced the more computationally viable ``top-m'' strategy \eqref{eq:Ihat_m} in Step 2. In this section, we provide theoretical guarantees of this strategy under certain conditions.
\begin{theorem} \label{cor:screening_m}
  Let 
  \begin{align}
    \I^\top_k = \left\{\ell \in [q]: \E(W_\ell^2) {\gamma_\ell^\ast}^2 \text{ is among the k largest} \right\}.
    \label{eq:I_k}
  \end{align}
  Under Assumption \ref{a:distn} and \ref{a:dim}, if $m \geq k$ and
  \begin{align}
    \min_{\ell \in \I^\top_k} \Psi_{\ell \ell}^{-1/2} |\Cov(Z_\ell, W^T \gamma^\ast)| \geq 
    \max_{\ell \notin \I^\top_k}\Psi_{\ell \ell}^{-1/2} |\Cov(Z_\ell, W^T \gamma^\ast)| + \eta^\ast
    \label{eq:signal_m}
  \end{align}
  where $\eta^\ast$ is in \eqref{eq:eta_star}, then
  \begin{align}
    \I^\top_k \subseteq \hat{\I}^\top_{m} 
    \label{eq:size_Ihat_m}
  \end{align}
  holds with probability greater than $1 - 8 p^{-2(\kappa^{3/5} - 1)} - 2 p^{-1} - \Prob(\Event_R^C)$. 
\end{theorem}
\begin{proof}
  Suppose that $|\omega_1| > |\omega_2| > ... > |\omega_q|$, where $|\omega_\ell| =\overline{\mathrm{sd}}(\r) |\overline{\mathrm{cor}}\left( \Z_\ell, \r \right)|$.
  Then $\hat{\I}_m = [m]$.
  For any $\ell \in \I^\top(k)$, by triangle inequality,
  \begin{align}
    |\omega_\ell| \geq |\omega_\ell^\ast| - |\omega_\ell - \omega_\ell^\ast| \nonumber
  \end{align}
  Now, let $h$ be the largest index such that $h \leq m$ and $h \notin \I^\top_k$. If such $h$ does not exist, then it must hold that $m = k$ and $\hat{\I}^\top_m = \I^\top_k$, and thus the result holds. If such $h$ exists, then
  \begin{align}
    |\omega_{m}| \leq |\omega_{h}| \leq |\omega_{h}^\ast| + |\omega_{h} - \omega_{h}^\ast|
    \leq \max_{h \notin \I^\top_k}|\omega_{h}^\ast| + |\omega_{h} - \omega_{h}^\ast|, \nonumber
  \end{align}
  which implies that
  \begin{align}
    |\omega_\ell| - |\omega_{m}| \geq & |\omega_\ell^\ast| - \max_{h \notin \I^\top_k}|\omega_{h}^\ast| - |\omega_\ell - \omega_\ell^\ast| - |\omega_{h} - \omega_{h}^\ast| \geq \min_{\ell \in \I^\top_k} |\omega_\ell^\ast| - \max_{\ell \notin \I^\top_k} |\omega_\ell^\ast| - 2 \max_\ell |\omega_\ell - \omega^\ast_\ell |.
    \nonumber
  \end{align}
  Then by assumption \eqref{eq:signal_m} and Lemma \ref{lem:screening}, we have that $|\omega_\ell| \geq |\omega_{m}|$ with certain probability, which implies that $\ell \in \hat{\I}^\top_m$.
\end{proof}

\section{Details of Section \ref{subsec:example}} \label{app:example}
\subsection{Gaussian case, with a single interaction} \label{app:gaussian}
We assume that $X \sim N(0, \Sigma)$ and there is only one true interaction, e.g., $\supp(\gamma^\ast) = \left\{ \tau(1, 2) \right\}$. We discuss the validity of the condition that $\eta(\bar{\alpha}) \geq \eta^\ast$ in Theorem \ref{thm:screening}.

Recall that in the Gaussian case, we have $W = Z$ and $\theta^\ast = \beta^\ast$. Without loss of generality, we assume that $\Sigma_{jj} = 1$ for all $j = 1, \dots, p$, so that for any pair of variable $X_j$ and $X_k$, their covariance $\sigma_{jk}$ equals their correlation coefficient $\rho_{jk}$. Furthermore, we have
\begin{align}
      &\E(Z_{\tau(j, k)}^2) = \E(X_j^2 X_k^2) = \sigma_{jj}^2 \sigma_{kk}^2 + 2 \sigma_{jk}^2 = 1 + 2 \rho_{jk}^2 \nonumber \\
      & \Var(Z_{\tau(j, k)}) = \E(Z_{\tau(j, k)}^2) - \E(Z_{\tau(j, k)})^2 = 1 + 2 \rho_{jk}^2 - \E(X_j \ast X_k)^2 = 1 + \rho_{jk}^2 \nonumber.
\end{align}
Also note that for any $(t, s) \in [p] \times [p]$,
\begin{align}
  \Cov(Z_{\tau(j, k)}, Z_{\tau(t, s)}) &= \E\left[ Z_{\tau(j, k)} Z_{\tau(t, s)} \right] - \E(Z_{\tau(j, k)}) \E (Z_{\tau(t, s)}) \nonumber\\
                                       &= \E\left( X_j X_k X_{t} X_s \right) - \E\left( X_j X_k \right) \E\left( X_{t} X_s \right) \nonumber\\
                                       &= \sigma_{jk} \sigma_{t s} + \sigma_{j t} \sigma_{ks} + \sigma_{js} \sigma_{k t} - \sigma_{jk} \sigma_{t s} \nonumber\\
                                       &= \rho_{jt}\rho_{ks} + \rho_{js} \rho_{k t}. \nonumber
\end{align}
With $\supp(\gamma^\ast) = \{(1, 2)\}$, for any $\A \subseteq [q]$, $W_{\A^C}^T \gamma^\ast_{\A^C} = X_1 X_2 \gamma^\ast_{\tau(1, 2)}$ if $\tau(1, 2) \notin \A$, and $W_{\A^C}^T \gamma^\ast_{\A^C} = 0$ if $\tau(1, 2) \in \A$. Recall that $\Sigma_{11} = \Sigma_{22} = 1$, we have
\begin{align}
  \E \left[ \exp \left( \frac{3|X_1 X_2 \gamma^\ast_{\tau(1, 2)}|}{8{\gamma^\ast_{\tau(1,2)}}}  \right) \right] 
  \leq &\E \left[ \exp \left( \frac{3 X_1^2 + 3 X_2^2}{16}\right) \right] \nonumber\\
  \leq & \frac{1}{2} \E \left[ \exp \left(\frac{3X_1^2}{8}\right) \right] + \frac{1}{2} \E \left[ \exp \left(\frac{3 X_2^2}{8}\right) \right] \leq 2. \nonumber
\end{align}
As a result, $\snorm{(W_{\A^C}^T \gamma^\ast_{\A^C})^2}_{\psi_{1/2}} \leq \snorm{W_{\A^C}^T \gamma^\ast_{\A^C}}_{\psi_{1}}^2 \leq \snorm{X_1X_2 \gamma^\ast_{\tau(1, 2)}}^2_{\psi_{1}} = \E[(W_{\A^C}^T \gamma^\ast_{\A^C})^2] = \frac{64}{9} {\gamma^\ast_{\tau(1,2)}}^2$.
On the other hand, by Lemma \ref{lem:exp} we have $(1 + 2 \rho_{jk}^2) {\gamma^\ast_{\tau(1,2)}}^2 = \E[(X_1X_2\gamma^\ast_{\tau(1, 2)})^2] \leq 4\snorm{(W_{\A^C}^T \gamma^\ast_{\A^C})^2}_{\psi_{1/2}}$.
Consequently, $ \frac{1}{4} {\gamma^\ast_{\tau(1,2)}}^2 \leq \snorm{(W_{\A^C}^T \gamma^\ast_{\A^C})^2}_{\psi_{1/2}} \leq \frac{64}{9} {\gamma^\ast_{\tau(1,2)}}^2$. And thus $\snorm{(W_{\A^C}^T \gamma^\ast_{\A^C})^2}_{\psi_{1/2}} = C {\gamma^\ast_{\tau(1, 2)}}^2$ for some constant $C \in [ \frac{1}{4}, \frac{64}{9} ]$.

If $C {\gamma^\ast_{\tau(1,2)}}^2 \leq \alpha$, then by definition \eqref{eq:I_alpha} we have $\I(\alpha) = \emptyset$. 
If $C {\gamma^\ast_{\tau(1,2)}}^2 > \alpha$, then any $\alpha$-important set of interactions should include $\tau(1, 2)$. By definition in \eqref{eq:I_alpha}, $\I(\alpha)$ should be the smallest set that contains $\tau(1, 2)$, which is $\{\tau(1, 2)\}$. In summary, 
\begin{align}
  \I(\alpha) &= \left\{ \tau(j, k) \in [q]: C {\gamma^\ast_{\tau(1,2)}}^2 > \alpha \right\} \nonumber\\
             &= \begin{cases}
               \left\{ \tau(1, 2) \right\} = \supp(\gamma^\ast) & \quad \text{if } \quad C {\gamma^\ast_{\tau(1, 2)}}^2 > \alpha \\
               \emptyset & \qquad \text{otherwise},
             \end{cases}
             \nonumber
\end{align}
and
\begin{align}
  \eta(\alpha) &= \frac{2}{3} \min_{\tau(j, k) \in \I(\alpha)} \frac{1}{\sqrt{1 + \rho_{jk}^2}} \left|\sum_{\tau(t, s) \in \supp(\gamma^\ast)} \left(  \rho_{j t} \rho_{ks} + \rho_{js} \rho_{k t} \right) \gamma^\ast_{\tau(t, s)} \right| \nonumber\\
               &= \begin{cases}
                 \frac{2}{3} \sqrt{1 + \rho_{12}^2} |{\gamma_{\tau(1, 2)}^\ast}|  & \qquad \text{if } \qquad C {\gamma^\ast_{\tau(1, 2)}}^2 > \alpha \\
                 \infty & \qquad \text{otherwise}.
               \end{cases}
               \nonumber
\end{align}

Next we give an upper bound on $\eta^\ast$. First note that, $\snorm{\diag(\Psi)^{-1/2} Z}_{\psi_1} = 1$, and from earlier discussion
\begin{align}
  \snorm{W^T \gamma^\ast}_{\psi_1} = \snorm{W_{\tau(1, 2)} \gamma^\ast_{\tau(1, 2)}}_{\psi_1} \leq 3 |\gamma^\ast_{\tau(1, 2)}|,
  \nonumber
\end{align}
and
\begin{align}
  \max_{\ell} \frac{|\Omega_{\ell}^T \gamma^\ast|}{\sqrt{\Psi_{\ell \ell}}} 
  = \max_{j,k}
  \frac{|\Omega_{\tau(j, k) \tau(1, 2)} \gamma^\ast_{\tau(1,2)}|}{\sqrt{\Psi_{\tau(j, k) \tau(j, k)}}} 
  = \max_{j, k}  \frac{|(\rho_{j1} \rho_{k2} + \rho_{j2} \rho_{k1}) \gamma^\ast_{\tau(1,2)}|}{\sqrt{1 + \rho_{jk}^2}} 
  \leq 2 |\gamma^\ast_{\tau(1, 2)}|.
  \nonumber
\end{align}
Furthermore, by the assumption that $\Sigma_{jj} = 1$, we have $\snorm{X}_{\psi_2} \leq 1$. From \eqref{eq:eta_star}, we have
\begin{align}
  \eta^\ast &\leq K \left[ \left(\snorm{W^T \gamma^\ast}_{\psi_1} + \max_{\ell} \frac{|\Omega_\ell^T \gamma^\ast|}{\sqrt{\Psi_{\ell \ell}}} \right)  \frac{(\log p)^{3/4}}{n^{1/2}} + \norm{\beta^\ast}_1^{1/2} \left( \sigma + 2|\gamma^\ast_{\tau(1, 2)}| \right)^{1/2} \left(\frac{\log p}{n} \right)^{1/4} +  \sigma \frac{(\log p)^{1/2}}{n^{1/2}} \right] \nonumber\\
            &\leq K \left[ 4 \frac{(\log p)^{3/4}}{n^{1/2}}  |\gamma^\ast_{\tau(1, 2)}| + \norm{\beta^\ast}_1^{1/2} \left( \sigma^{1/2} + \sqrt{2}|\gamma^\ast_{\tau(1, 2)}|^{1/2} \right) \left(\frac{\log p}{n} \right)^{1/4} + \sigma \frac{(\log p)^{1/2}}{n^{1/2}} \right].
            \nonumber
\end{align}

By Assumption \ref{a:dim} that $\kappa \log p \leq n^{1/2}$ for some constant $\kappa > 1$, we have that
\begin{align}
  \frac{2}{3} \sqrt{1 + \rho_{12}^2} - 4K \frac{(\log p)^{3/4}}{n^{1/2}} \geq \frac{2}{3}  - \frac{1}{3} \frac{12 K}{n^{1/8}} \frac{(\log p)^{3/4}}{n^{3/8}} 
  \geq \frac{2}{3}  - \frac{1}{3} \frac{12 K }{n^{1/8}\kappa^{3/4}}  \geq \frac{1}{3} \nonumber,
\end{align}
for $n \geq (12K)^8 \kappa^{-6}$.
For the condition that $\eta(\bar{\alpha}) \geq \eta^\ast$ to hold for some $\bar{\alpha}$, it is sufficient to require that
\begin{align}
    &\frac{1}{3} |\gamma^\ast_{\tau(1,2)}| \geq \left(\frac{1}{3} \sqrt{1 + \rho_{12}^2} - 4K \frac{(\log p)^{3/4}}{n^{1/2}}  \right) |{\gamma_{\tau(1, 2)}^\ast}| \nonumber\\
  \geq & K \left[ \sqrt{2C} \snorm{\beta^\ast}_1^{1/2} \left( \frac{\log p}{n}  \right)^{1/4} |\gamma^\ast_{\tau(1, 2)}|^{1/2} + \snorm{\beta^\ast}_1^{1/2} \sqrt{\sigma} \left( \frac{\log p}{n}  \right)^{1/4} + \sigma \left(\frac{\log p}{n}\right)^{1/2} \right].
  \label{eq:quadratic}
\end{align}
A sufficient condition for \eqref{eq:quadratic}, and thus \eqref{eq:eta_star}, to hold is that the signal strength is large enough, i.e., $|\gamma^\ast_{\tau(1, 2)}| \geq r(n, q)^2$, where
\begin{align}
 & r(n, q) := \left[ 18 C K^2 \snorm{\beta^\ast}_{1} \left( \frac{\log p}{n}  \right)^{1/2} + 12K  \snorm{\beta^\ast}_1^{1/2} \sigma^{1/2} \left( \frac{\log p}{n}  \right)^{1/4} + 12K \sigma \left( \frac{\log p}{n}  \right)^{1/2}   \right]^{1/2} \nonumber\\
 &\geq \frac{3K \sqrt{2} }{2} \snorm{\beta^\ast}_1^{1/2} \left( \frac{\log p}{n}  \right)^{1/4} + \left[ \frac{9 K^2}{2} \snorm{\beta^\ast}_1 \left( \frac{\log p}{n}  \right)^{1/2} + 3K \left( \snorm{\beta^\ast}_1^{1/2} \sigma^{1/2} \left( \frac{\log p}{n}  \right)^{1/4} + \sigma \left(\frac{\log p}{n} \right)^{1/2}  \right)  \right]^{1/2} , \nonumber
\end{align}
and the right hand side of the inequality above is the smallest value of $|\gamma^\ast_{\tau(1, 2)}|^{1/2}$ that satisfies \eqref{eq:quadratic}.

\subsection{Independent Bernoulli case} \label{app:bernoulli}
We now consider the case where $\Prob(X_j = 1) = p_j$ and $\Prob(X_j = 0) = 1 - p_j$ for each $j$. And $X_j$'s are independent. Then
\begin{align}
      &\E(Z_{\tau(j, k)}) = \E(X_j X_k) = \begin{cases}
        p_j \qquad &j = k \\
        p_j p_k \qquad &j \neq k
      \end{cases} \nonumber\\
      &\E(Z_{\tau(j, k)}^2) = \E(X_j^2 X_k^2) = \begin{cases}
        p_j \qquad &j = k \\
        p_jp_k \qquad &j \neq k
      \end{cases}\nonumber\\
      & \Sigma_{jk} = \begin{cases}
        p_j (1 - p_j) \qquad & j = k\\
        0 \qquad & j \neq k.
      \end{cases}
      \nonumber
\end{align}

Without loss of generality, assume $t \leq s$. Note that
\begin{align}
  \Phi_{j, \tau(t, s)} =  \Cov(X_j, Z_{\tau(t, s)}) 
    &= \E\left[ X_j Z_{\tau(t, s)} \right] - \E(X_j) \E (Z_{\tau(t, s)}) \nonumber\\
    &= \E\left( X_j X_{t} X_s \right) - \E\left( X_j \right) \E\left( X_{t} X_s \right) \nonumber\\
    &= \begin{cases}
      0 \qquad &j < t < s \\
      p_tp_s(1 - p_t) \qquad &t = j < s \\
      0 \qquad &t < j < s \\
      p_tp_s(1 - p_s) \qquad &t < j = s \\
      0 \qquad &t < s < j\\
      0 \qquad &t = s < j\\
      0 \qquad &t = s > j\\
      p_j (1 - p_j) \qquad &t = s = j. \\
    \end{cases}
\end{align}
Similarly that for any $(t, s) \in [p] \times [p]$,
\begin{align}
  \Psi_{\tau(j,k), \tau(t, s)} = \Cov(Z_{\tau(j, k)}, Z_{\tau(t, s)}) 
    &= \E\left[ Z_{\tau(j, k)} Z_{\tau(t, s)} \right] - \E(Z_{\tau(j, k)}) \E (Z_{\tau(t, s)}) \nonumber\\
    &= \E\left( X_j X_k X_{t} X_s \right) - \E\left( X_j X_k \right) \E\left( X_{t} X_s \right). \nonumber
\end{align}

Now for any $(t, s) \in [p] \times [p]$,
\begin{align}
  W_{\tau(t, s)} &= Z_{\tau(t, s)} - \sum_{j = 1}^p \Phi_{j, \tau(t,s)} \frac{X_j}{\Sigma_{jj}}
  = Z_{\tau(t, s)} - \sum_{j = 1}^p \Phi_{j, \tau(t, s)} \frac{X_j}{p_j (1 - p_j)} \nonumber\\
                 &= 
                 \begin{cases}
                   Z_{\tau(t, t)} - X_t \qquad & t = s \\
                   Z_{\tau(t, s)} - (p_s X_t + p_t X_s)  \qquad & t \neq s. 
                 \end{cases}
\end{align}
For simplicity assume that there is only one interaction, e.g., $\supp(\gamma^\ast) = \{\tau(1, 2)\}$.
For any $\A \subseteq [q]$, $W_{\A^C}^T \gamma^\ast_{\A^C} = (X_1X_2 - p_2 X_1 - p_1 X_2) \gamma^\ast_{\tau(1, 2)}$ if $\tau(1, 2) \notin \A$, and $W_{\A^C}^T \gamma^\ast_{\A^C} = 0$ if $\tau(1, 2) \in \A$. For any value $B > 0$, 
\begin{align}
    &\E \left[ \exp \left( \frac{(X_1X_2 - p_2 X_1 - p_1 X_2)^2}{B}\right) \right] \nonumber\\
  = & p_1 p_2 \exp \left( \frac{(1 - p_1 - p_2)^2}{B} \right) + p_1 (1 - p_2) \exp \left( \frac{p_2^2}{B}  \right) + p_2 (1 - p_1) \exp \left( \frac{p_1^2}{B}  \right) + (1 - p_1)(1 - p_2). \nonumber
\end{align}
Then by definition of sub-Weibull(1/2) norm, $\snorm{(W_{\A^C}^T \gamma^\ast_{\A^C})^2}_{\psi_{1/2}} = C {\gamma^\ast_{\tau(1,2)}}^2$ for some constant $C$ that is the smallest value such that the above equation is bounded by 2.
With similar discussion as in \ref{app:gaussian}, we have
\begin{align}
  \I = \begin{cases}
    \supp(\gamma^\ast) = \{\tau(1, 2)\} &\quad \text{if} \quad C{\gamma^\ast_{\tau(1, 2)}}^2 > \alpha \nonumber\\
    \emptyset &\quad \text{otherwise}.
  \end{cases}    
\end{align}
Note from \eqref{eq:eta_alpha} that
\begin{align}
  \eta (\alpha) &= \frac{2}{3} \Psi_{\tau(1, 2) \tau(1, 2)}^{-1/2} |\Cov(Z_{\tau(1, 2)}, W^T \gamma^\ast)| \nonumber\\
                &= \frac{2}{3} \Psi_{\tau(1, 2) \tau(1, 2)}^{-1/2} |\Cov(Z_{\tau(1, 2)}, W_{\tau(1, 2)} )||\gamma^\ast_{\tau(1, 2)}|.
                \nonumber
\end{align}
We have
\begin{align}
  \Psi_{\tau(1, 2)\tau(1, 2)} 
  = \E \left( X_1 X_2 X_1 X_2 \right) - \E \left( X_1 X_2 \right) \E \left( X_1 X_2 \right) 
  = p_1 p_2 (1 - p_1 p_2),
  \nonumber
\end{align}
and
\begin{align}
  \Cov \left( Z_{\tau(1, 2)}, W_{\tau(1, 2)} \right) 
    &= \Cov \left( Z_{\tau(1, 2)}, Z_{\tau(1, 2)} \right) - p_1 \Cov \left( Z_{\tau(1, 2)}, X_2 \right) - p_2 \Cov \left( Z_{\tau(1, 2)}, X_1 \right) \nonumber\\
    &=p_1 p_2 (1 - p_1 p_2) - p_1^2 p_2 (1 - p_2) -  p_1 p_2^2 (1 - p_1) \nonumber\\
    &=p_1 p_2 \left( 1 + p_1 p_2  - p_1 - p_2 \right).
    \nonumber
\end{align}
Therefore, we have
\begin{align}
  \eta(\alpha) = \begin{cases}
    \frac{2 |1 + p_1 p_2 - p_1 - p_2|}{3(1 - p_1 p_2)} |\gamma^\ast_{\tau(1, 2)}| & \qquad \text{if} \quad C{\gamma^\ast_{\tau(1, 2)}}^2 > \alpha \\
    \infty & \qquad \text{otherwise}.
  \end{cases}
  \nonumber
\end{align}
Then we could follow the same discussion as in Appendix \ref{app:gaussian}.

\end{document}